\documentclass[10pt, a4paper]{article}

\usepackage[affil-it]{authblk} 
\usepackage{etoolbox}
\usepackage{lmodern}

\makeatletter
\patchcmd{\@maketitle}{\LARGE \@title}{\fontsize{18}{19.2}\selectfont\@title}{}{}
\makeatother

\usepackage{setspace}
\usepackage{authblk}
\usepackage{multirow}
\usepackage{amsmath}
\usepackage{pdfpages}
\usepackage{graphicx}
\usepackage{physics}
\usepackage{amsmath}
\usepackage{epstopdf}
\usepackage{amssymb}
\usepackage{mathtools}
\usepackage{amsthm}
\usepackage[toc,page]{appendix}
\usepackage{float}
\usepackage{framed} 
\usepackage{algorithmic}
\usepackage{cancel}
\usepackage[ruled,vlined]{algorithm2e}
\usepackage[colorlinks]{hyperref}
\AtBeginDocument{%
  \hypersetup{
    citecolor=blue,
    linkcolor=blue,
    urlcolor=blue}}
\input{Qcircuit}
\DeclarePairedDelimiter{\ceil}{\lceil}{\rceil}

\newcommand{\Hil}{\mathcal{H}}
\newcommand{\HilD}{\mathcal{H}^D}
\newcommand{\Hild}{\mathcal{H}^d}

\newcommand{\Hilin}{\mathcal{H}_{in}}
\newcommand{\Hilout}{\mathcal{H}_{out}}

\newcommand{\E}{\mathcal{E}}
\newcommand{\K}{\mathcal{K}}
\newcommand{\X}{\mathcal{X}}
\newcommand{\Y}{\mathcal{Y}}

\newcommand{\T}{\mathcal{T}}

\newcommand{\A}{\mathcal{A}}
\newcommand{\C}{\mathcal{C}}
\newcommand{\Sc}{\mathcal{S}}
\newcommand{\Ue}{\mathrm{U_{\mathcal{E}}}}
\newcommand{\U}{\mathrm{U}}
\newcommand{\Uset}{\mathcal{U}}
\newcommand{\Gen}{\mathrm{Gen}}
\newcommand{\mbraket}[2]{\bra{#1}#2\rangle} 
\newcommand{\negl}{negl}
\newcommand{\nonnegl}{non\text{-}\negl}

\newcommand{\kc}{\ket{\phi^c_i}}
\newcommand{\kr}{\ket{\phi^r_i}}
\newcommand{\krm}{\ket{\phi^r_i}^{\otimes M}}
\newcommand{\kp}{\ket{\phi^p_i}}

\newtheorem{theorem}{Theorem}
\newtheorem{lemma}{Lemma}

\newtheorem{definition}{Definition}
\newtheorem{game}{Game}
\usepackage[a4paper, total={6in, 8.6in}]{geometry}

\begin{document}

\title{On the Connection Between Quantum Pseudorandomness and Quantum Hardware Assumptions}
\date{}
\author[1]{Mina Doosti\thanks{m.doosti@sms.ed.ac.uk}}
\author[1]{Niraj Kumar}
\author[1,2]{Elham Kashefi}
\affil[1]{School of Informatics, 10 Crichton St., University of Edinburgh, United Kingdom}
\affil[2]{CNRS, LIP6, Sorbonne Universit\'{e}, 4 place Jussieu, 75005 Paris, France}
\author[1]{Kaushik Chakraborty\thanks{kchakrab@exseed.ed.ac.uk}}


\maketitle
\medskip

\begin{abstract}
This paper, for the first time, addresses the questions related to the connections between quantum pseudorandomness and quantum hardware assumptions, specifically quantum physical unclonable functions (qPUFs). Our results show that efficient pseudorandom quantum states (PRS) are sufficient to construct the challenge set for universally unforgeable qPUFs, improving the previous existing constructions based on the Haar-random states. We also show that both the qPUFs and the quantum pseudorandom unitaries (PRUs) can be constructed from each other, providing new ways to obtain PRS from the hardware assumptions. Moreover, we provide a sufficient condition (in terms of the diamond norm) that a set of unitaries should have to be a PRU in order to construct a universally unforgeable qPUF, giving yet another novel insight into the properties of the PRUs. Later, as an application of our results, we show that the efficiency of an existing qPUF-based client-server identification protocol can be improved without losing the security requirements of the protocol.  
\end{abstract}
\newpage
\maketitle

\section{Introduction} 
\label{sec:intro}

Pseudorandomness is one of the most fundamental concepts in the domain of cryptography and complexity theory. In contrast to true randomness, it captures the notion of primitives that behaves randomly to the computationally-bounded observers \cite{yao1982theory,shamir1983generation,blum1984generate}. The pseudorandom objects like pseudorandom number generators (PRGs) and pseudorandom functions (PRFs) play a crucial role in designing classical symmetric key cryptography protocols for secure communication \cite{goldreich1986construct,haastad1999pseudorandom, goldreich1984cryptographic, luby1988construct, rompel1990one}. These pseudorandom objects can be designed by exploiting the algebraic properties of the families of keyed functions like the keyed hash functions or from some hardware assumptions like the physical unclonable functions (PUFs). In the classical world, the relationship between PUF and pseudorandomness is well-studied~\cite{ruhrmair2009foundations}. Similar to classical pseudorandomness, recently Ji, Liu, and Song \cite{ji2018pseudorandom} introduced the concept of quantum pseudorandomness such as pseudorandom quantum states (PRSs) and pseudorandom unitaries (PRUs) as families of states or unitary transformations that are indistinguishable from Haar measure (true random measure) to any quantum computationally-bounded observers. On the other hand, similar to the classical PUFs, recently, we have the concept of the quantum PUFs \cite{arapinis2021quantum}. However, unlike the classical case, the relation between the quantum pseudorandomness and the quantum PUFs is not well-explored. The existing PRS schemes are constructed under computational assumptions such as quantum-secure PRF or quantum secure one-way functions~\cite{brakerski2020scalable,ji2018pseudorandom}. An interesting question that arises is whether quantum pseudorandomness can be achieved under a different set of assumptions? In this paper, to the best of our knowledge, for the first time, we show the construction of quantum pseudorandom unitaries from quantum PUFs and vice-versa. 

Quantum pseudorandom states are an ensemble of a keyed family of quantum states $\{|\phi_k\rangle\}_{k \in \K}$ that can be generated efficiently \cite{ji2018pseudorandom}. The pseudorandomness comes from the property that for any polynomial-time quantum adversary, any polynomial number of copies of the states that are sampled from the ensemble $\{|\phi_k\rangle\}_{k \in \K}$ is indistinguishable from the same number of copies of Haar-random states. Similarly, pseudorandom unitaries are an ensemble of a keyed family of unitaries $\{U_k\}_{k \in \K}$ that can be implemented efficiently \cite{ji2018pseudorandom}. Analogous to the PRS, the pseudorandomness of the PRUs implies that with oracle access to the unitary, no polynomial-time quantum adversary can distinguish between the unitaries that are sampled from $\{U_k\}_{k \in \K}$ and the Haar-measure unitaries.   

A PUF is designed to be a cost-efficient, and low resource hardware device that provides a unique physically defined digital fingerprint \cite{delvaux2017security, herder2014physical}. Corresponding to a challenge it produces a unique response that acts as an identifier. In the case of classical PUFs, the uniqueness comes from the unique physical variations that occur naturally during the manufacturing process of the device. Such subtle physical variations of the hardware components during the manufacturing process can be easily measured but are infeasible to reproduce in practice. It is well known in the classical setting that theoretically, PUFs can be considered pseudorandom functions. However, in practice, most of the classical PUFs are vulnerable to machine learning-based attacks \cite{ganji2016strong,ruhrmair2010modeling,khalafalla2019pufs}. Due to this shortcoming, there is a significant interest in designing quantum PUFs that utilise quantum mechanical properties (qPUFs), where both the challenges and responses are quantum states \cite{arapinis2021quantum, gianfelici2020theoretical,nikolopoulos2017continuous}. Quantum challenges and responses feature an additional property that they cannot be cloned by the laws of quantum mechanics \cite{wootters1982single}, in contrast to the classical case.

In general, a qPUF is modelled as a completely positive trace preserving (CPTP) map, which maps an input challenge state to a unique response state. In addition, a qPUF must also be unique i.e. two distinct qPUFs must generate different responses to any given challenge with high probability (\emph{uniqueness} property), and it must be unforgeable by any bounded (quantum or classical) adversary trying to clone the device (\emph{unforgeability property}). In \cite{arapinis2021quantum}, Arapinis et al. developed a formal security notion for the qPUFs and provided a qPUF construction based on Haar-random unitaries with the challenge states also being drawn from the Haar-random state, to satisfy the unforgeability property. Moreover, in the same paper, the authors design a generic quantum emulation-based attack for forging any qPUF and proved that their generic construction is unforgeable against any polynomial-time quantum adversary. This construction, although secure, is not practical due to the Haar-random requirement on the unitaries and the states. The reason for this is the fact that sampling from Haar measure requires exponential resources~\cite{knill1995approximation} and hence is experimentally challenging~\cite{carolan2015universal}. The construction of the unitary qPUF itself was partially improved by the result of Kumar et. al \cite{kumar2021efficient} where they constructed a qPUF based on unitary $t$-designs which are efficiently built. However, they still require the challenges to be drawn from the Haar-random set of states to prove the unforgeability property. Further, we emphasise that, unlike the classical PUFs, the literature on qPUFs is not yet mature, and we have only very few candidate designs for qPUFs as mentioned above.  

In this work, we make substantial progress in the above inefficient requirement of the challenges being chosen from the Haar measure. Specifically, we show that PRS can help reduce the challenger's overhead significantly in choosing the challenge states - from inefficient Haar-random states to efficient PRS. We further show that PRUs can be used as a viable candidate for qPUFs. This result provides yet another novel and efficient technique for constructing qPUFs.
Moreover, here we also investigate, whether qPUFs can be used as PRUs. Similar to the qPUF, PRU is also a relatively new concept, and to the best of our knowledge, there are no concrete designs for the PRUs. Our investigation in this paper helps establish a close connection between these two new fields, i.e., qPUFs and quantum pseudorandomness. This relation gives us novel insights into designing both qPUFs and the PRUs. We are optimistic that the connections that we foster here would benefit both communities and the advances in one field would help to enrich the advances in the other.
In the next subsection, we give a brief outline of our results. 

\subsection{Result overview}

In this paper, we first address the inefficiency issue of the qPUF design in \cite{arapinis2021quantum, kumar2021efficient} and prove the security against quantum polynomial-time adversaries even if the challenge states are sampled from a set of pseudorandom quantum states.

\begin{theorem}[informal]\label{th:inf-prs-unf}
Any unitary qPUF satisfies unforgeability with the challenges that are selected from a Pseudorandom Quantum States (PRS) family (instead of Haar random states) against quantum polynomial-time (QPT) adversaries. 
\end{theorem}

Here we also show that the PRUs can be used as qPUFs. Moreover, we establish a connection between a unitary family of qPUFs, with a specific practical requirement, and PRUs. 

\begin{theorem}[informal]\label{th:inf-pru-uu}
Any PRU family can be a unitary qPUF family.
\end{theorem}

\begin{theorem}[informal]\label{th:inf-practicaluu-pru}
A family of practically unknown unitaries is also a PRU family. Hence any unitary qPUF family that satisfies practical unknownness, is also a PRU family.
\end{theorem}
Later, we give a novel construction of PRUs from the family of qPUFs by exploring yet another hardware requirement which is their uniqueness property. The following result can also be applied to any unitary family with almost-maximal uniqueness, not only qPUFs.

\begin{theorem}[informal]\label{th:inf-unique-pru}
Any family of unitary transformation of over $d$-dimensional Hilbert space satisfying the almost-maximal uniqueness in the diamond norm is also a PRU family for sufficiently large $d$. Hence any PUF family satisfying this degree of uniqueness, is also a PRU.
\end{theorem}

The demographic summary of the above theorems can be found in Figure~\ref{fig:results}. We finish our paper with a secure and efficient qPUF-based client verification protocol using our result from Theorem $1$. 

\begin{theorem}[informal]\label{th:inf-protocols}
The qPUF-based identification protocols in~\cite{doosti2020client} can achieve the same security guarantee against QPT adversary if the Haar-random states are replaced with PRS.
\end{theorem}

\begin{figure}
\includegraphics[scale=0.35]{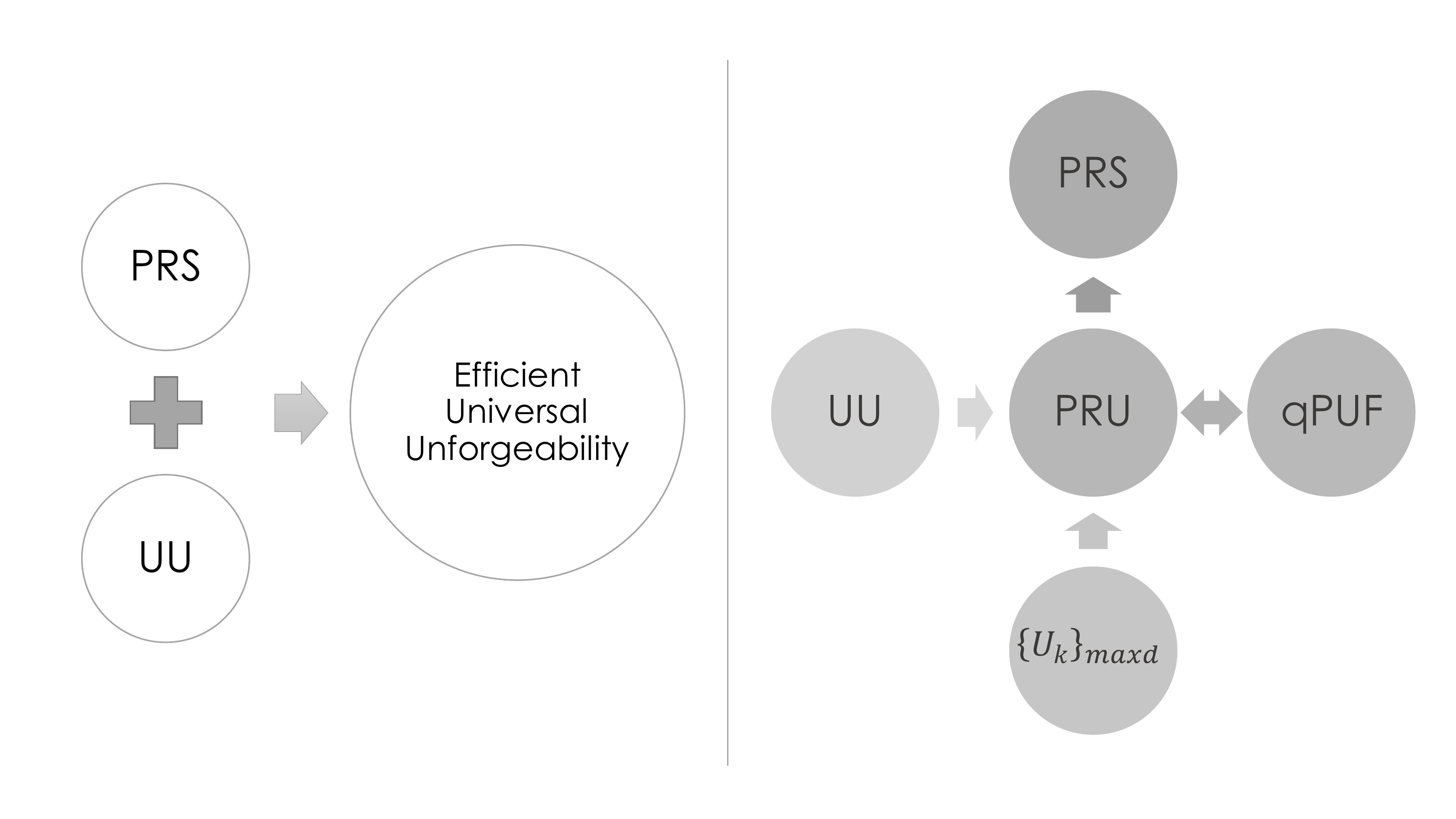}
    \centering
    \caption{Demographic Summary of the results. The left hand figure demonstrates Theorem~\ref{th:efficientuu-prs} stating that the universal unforgeability of unknown unitaries can be achieved efficiently using PRS. The right figure depicts the relationship between unknwon unitaries (UU), quantum physical unclonable functions (qPUF), pseudorandom unitaries (PRU) and families of almost maximally-distanced unitaries ($\{U_k\}_{maxd}$) proved in the Theorems~\ref{th:pru-uu}, \ref{th:pru-unique}, \ref{th:practicaluu-pru}, and \ref{th:max-unique-pru}. It also shows that they can be used as generators for pseudorandom quantum states (PRS).}
    \label{fig:results}
\end{figure}

\subsection{Notations}
Here we provide some of the widely-used notations in this paper. 
\begin{itemize}
\small
    \item PRG : Pseudorandom generator.
    \item PRF : Pseudorandom function.
    \item qPRF : Quantum-Secure Pseudorandom function.
    \item PRS : Pseudorandom state.
    \item PRU : Pseudorandom unitary.
    \item qPUF : Quantum physical unclonable function.
    \item QPT : Quantum polynomial-time.
    \item UU : Unknown Unitary transformation.
    \item CRP : Challenge response pair.
    \item BQP : Bounded quantum polynomial.
    \item CPTP : Completely positive trace preserving.
    \item $F(.,.)$ : Uhlmann's fidelity.
    \item $\mu$ : Haar measure.
    \item $\lambda$ : Security parameter.
\end{itemize}
\section{Preliminaries} \label{sec:prelims}

This section presents the various ingredients required for our results and proofs.

\subsection{Quantum Pseudorandomness}\label{sec:prelim-qpseudorand}
Pseudorandomness is a central concept in modern cryptography which has also been extended to the quantum regime. Here we mention different notions that have been defined or extended into the quantum world namely, Pseudorandom Quantum States (PRS), quantum-secure Pseudorandom Functions (qPRF) and their quantum analogue, namely quantum Pseudorandom Unitaries (PRU). 

\subsubsection{Pseudorandom Quantum States (PRS)}
\begin{definition}\label{def:prs}[Pseudorandom Quantum States(PRS): \cite{ji2018pseudorandom}]
Let $\Hil$ be a Hilbert space and $\K$ the key space. $\Hil$ and $\K$ depend on the security parameter $\lambda$. A keyed family of quantum states $\{\ket{\phi_k}\in S(\Hil)\}_{k\in\K}$ is \textit{pseudorandom}, if the following two conditions hold:
\begin{itemize}
    \item \textbf{Efficient generation}. There is an efficient quantum algorithm $G$ which generates the state $\ket{\phi_k}$ on input $k$. That is, for all $k\in\K, G(k) = \ket{\phi_k}$.
    \item \textbf{Pseudorandomness}. Any polynomially many copies of $\ket{\phi_k}$ with the same random $k\in\K$ is computationally indistinguishable from the same number of copies of a Haar random state. More precisely, for any efficient quantum algorithm $\A$ and any $m\in poly(\lambda)$, 
\end{itemize}
\begin{equation}
    |\underset{k \leftarrow \K}{Pr}[\A(\ket{\phi_k}^{\otimes m})=1] - \underset{\ket{\psi} \leftarrow \mu}{Pr}[\A(\ket{\psi}^{\otimes m})=1]| = \negl(\lambda).
\end{equation}
where $\mu$ is the Haar measure on $S(\Hil)$.
\end{definition}

\subsubsection{Quantum-secure Pseudorandom Function (qPRF)}
Quantum-secure pseudorandom functions are families of functions that look like truly random functions to QPT adversaries. Formally, qPRF are defined as follows:

\begin{definition}\label{def:qprf}[Quantum-Secure Pseudorandom Functions(qPRF): \cite{ji2018pseudorandom}]
Let $\K,\X,\Y$ be the keyspace, the domain and range, all implicitly depending on the security parameter $\lambda$. A keyed family of functions $\{PRF_k: \X \rightarrow \Y\}_{k\in \K}$ is a quantum-secure pseudorandom function (qPRF) if for any polynomial-time quantum oracle algorithm $\A$, $PRF_k$ with a  random $k \leftarrow \K$ is indistinguishable from a truly random function $f \leftarrow \Y^{\X}$ in the sense that:
\begin{equation}
|\underset{k \leftarrow \K}{Pr}[\A^{PRF_k}(1^{\lambda})=1] -\underset{f \leftarrow \Y^{\X}}{Pr}[\A^{f}(1^{\lambda})=1]| = \negl(\lambda).
\end{equation}
\end{definition}

\subsubsection{Pseudorandom Unitary Operators (PRUs)}
These are unitary equivalent of PRFs defined as follows.
\begin{definition}\label{def:pru}[Pseudorandom Unitary Operators(PRU): \cite{ji2018pseudorandom}]
A family of unitary operators $\{U_k \in \mathcal{U}(\Hil)\}_{k \in \mathcal{K}}$ is a pseudorandom unitary if two conditions hold:
\begin{itemize}
    \item \textbf{Efficient computation}. There is an efficient quantum algorithm $Q$ such that for all $k$ and any state $\ket{\psi} \in S(\Hil), Q(k,\ket{\psi}) = U_k\ket{\psi}$.
    \item \textbf{Pseudorandomness}. $U_k$ with a random key $k$ is computationally indistinguishable from a Haar random unitary operator. More precisely, for any efficient quantum algorithm $\A$ that makes at most polynomially many queries to the oracle:  
\end{itemize}
\begin{equation}
    |\underset{k \leftarrow \K}{Pr}[\A^{U_k}(1^{\lambda})=1] - \underset{U \leftarrow \mu}{Pr}[\A^U(1^{\lambda})=1]| = \negl(\lambda).
\end{equation}
where $\mu$ is the Haar measure on $S(\Hil)$. Note that here we focus on the Pseudorandomness condition of the PRU definition.
\end{definition}

\subsubsection{Unknown Unitary Transformations (UUs)}
We also mention a relevant notion to PRU, called a family of Unknown Unitaries (UU) defined in~\cite{arapinis2021quantum}, that can also be interpreted as single-shot pseudorandomness.

\begin{definition}[Unknown Unitary Transformation~\cite{arapinis2021quantum}]\label{def:uu} We say a family of unitary transformations $U^u$, over a $d$-dimensional Hilbert space $\Hild$ is called Unknown Unitaries if for all QPT adversaries $\A$ the probability of estimating the output of $U^u$ on any state $\ket{\psi}\in\Hild$ selected uniformly from Haar measure, is at most negligibly higher than the probability of estimating the output of a Haar random unitary operator on that state:
\begin{equation}\small
        |\underset{U \leftarrow U^u}{Pr}[F(\A(\ket{\psi}),U\ket{\psi}) \geq \delta(\lambda)] - \underset{U_{\mu} \leftarrow \mu}{Pr}[F(\A(\ket{\psi}),U_{\mu}\ket{\psi}) \geq \delta(\lambda)]| = \negl(\lambda).
\end{equation}
where $\delta(\lambda)$ is any non-negligible function in the security parameter.
\end{definition}

We note that this definition characterises a notion of \emph{single-shot} indistinguishability from the family of Haar-random unitaries. Thus the adversary has only a single black-box query access to the unitary, but can have some existing prior information from the family that can be used for estimating the output. The definition intuitively states that a family of unitary is unknown when no such useful information exist about the family prior to the query access. 

%
\subsection{Quantum Adversarial Model and Security Definitions}
Strong notions of the security for quantum cryptographic proposals require cryptanalysis against adversaries which also possess quantum capabilities of varying degrees \cite{boneh2011random, mosca2018cybersecurity, song2014note}.
The strongest of such notions is achieved by assuming no restrictions on the adversary's computational power and resources. This security model, also known as security against \emph{unbounded adversary}, is usually too strong to be achieved by most cryptographic primitives such as qPUFs.
It has been shown in~\cite{arapinis2021quantum}, that unitary qPUFs cannot remain secure against an unbounded adversary. Thus the standard security model that we also use in this paper is the notion of security against efficient quantum adversaries, or in other words, QPT adversaries. We define such an adversary in the context of qPUFs. A QPT adversary with query access to qPUF is defined as an adversary that can query the qPUF oracle with polynomially many (in the security parameter) arbitrary challenges and has a polynomial sized quantum register to store the quantum CRPs. The QPT adversary is also allowed to run any efficient quantum algorithm. The security of most qPUF-based cryptographic protocols relies on the unforgeability property of qPUF.

Here we follow the same definitions of \emph{universal} unforgeability  (Also called \emph{selective unforgeability} in the context of quantum PUFs) defined in~\cite{arapinis2021quantum,doosti2021unified} and restate them as follows:

\begin{game}\label{game:uni-unf}[Universal Unforgeability] Let $\Gen$, $\Ue$ and $\T$ be the generation, evaluation and test algorithms of the quantum primitive $\E$ respectively. We define the following game $G$ running between an adversary $\A$ and a challenger $\C$:
\begin{itemize}
    \item [] \textbf{Setup phase.} The challenger $\C$ runs $\Gen(\lambda)$ and reveals to the adversary $\A$, the domain and range Hilbert space of $\Ue$ respectively denoted by $\Hilin$ and $\Hilout$

    \item [] \textbf{Learning phase.} For $i=1:k$
        \begin{itemize}
            \item $\A$ issues arbitrary query $\rho_i \in \Sc(\Hilin)$ to $\C$;
            \item $\C$ generates $\rho_i^{out} = \Ue\rho_i\Ue^{\dagger}$ and sends $\rho_i^{out}$, to $\A$;
        \end{itemize} 
    \item [] \textbf{Challenge phase.}
        $\C$ chooses a quantum state $\rho^*$ at random from the uniform distribution (Haar) over the Hilbert space $\Hilin$ and sends $\rho^*$ to $\A$. The challenger can generate arbitrary copies of $\rho^*$.
    \item [] \textbf{Guess phase.} 
    \begin{itemize}
        \item $\A$ generates the forgery $\omega$ and sends it to~ $\C$;

        \item $\C$ runs the test algorithm $b\leftarrow \mathcal{T}((\rho*^{out})^{\otimes \kappa_1},\omega)$
        where $b\in\{0,1\}$ and outputs $b$. The adversary wins the game if $b=1$.
    \end{itemize}
\end{itemize}
\end{game}

\begin{definition}[Quantum Universal Unforgeability]\label{def:qunf}
A primitive provides quantum universal unforgeability if the success probability of any QPT adversary $\A$ in winning the Game~\ref{game:uni-unf} is negligible in the security parameter $\lambda$
\begin{equation}
Pr[1\leftarrow G(\lambda, \A)] = \negl(\lambda)
\end{equation}
\end{definition}

Throughout the paper we widely use the result from~\cite{doosti2021unified,arapinis2021quantum} implying that unknown unitary transformations as formalized by Definition~\ref{def:uu}, can satisfy the notion of universal Unforgeability:

\begin{theorem}[\cite{doosti2021unified}]\label{th:uu-unforge}
Primitives with their evaluation algorithm being an unknown unitary transformation are universally unforgeable.
\end{theorem}

\subsection{Quantum Equality Tests}\label{sec:test}
Distinguishing two unknown quantum states is a central ingredient in quantum information processing. This task is often referred to as the ``state discrimination task''. The celebrated Holevo-Helstrom bound \cite{holevo1973bounds} relates the optimal distinguishability of two unknown states with the trace distance between the density matrices. This implies that unless the states are the same (up to a global factor), it is impossible to deterministically distinguish the two states. An important application of state discrimination is the task of Equality testing \cite{buhrman2001quantum, barenco1997stabilization, xu2015experimental}. This is an extremely simple task but a building block for lots of complicated quantum protocols. The objective of Equality testing, one that we consider in our work, is to test whether two \emph{unknown} quantum states are the same. This is a well-studied topic and we describe the optimal quantum protocols for Equality testing.

\subsubsection{SWAP test}\label{sec:swap}

Given a single copy of two unknown quantum states $\rho$ and $\sigma$, is there a simple test to optimally determine whether the two states are equal or not? This question was answered in affirmative by Buhrman et al \cite{buhrman2001quantum} when they provided a test called the SWAP test. This test was initially used by the authors to prove an exponential separation between classical and quantum resources in the simultaneous message passing model. Since then it has been used as a standard tool in the design of various quantum algorithms \cite{buhrman2010nonlocality,kumar2017efficient}. A SWAP test circuit takes as an input the two unknown quantum states $\rho$ and $\sigma$ and attaches an ancilla $\ket{0}$. A Hadamard gate is applied to the ancilla followed by the control-SWAP gate and again a Hadamard on the ancilla qubit. Finally, the ancilla is measured in the computational basis and we conclude that the two states are equal if the measurement outcome is `0' (labelled accept). Figure~\ref{fig:swap} illustrates this test in the special case when the state $\sigma$ is a pure state and shown by $\ket\psi$.
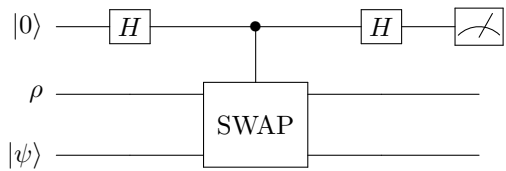
\begin{figure}[h!]
    \centering
    \[    \Qcircuit @C=2em @R=1.4em {
       & \lstick{\ket{0}} & \gate{H} & \ctrl{1} & \gate{H} & \meter \\
       & \lstick{\rho} & \qw & \multigate{1}{\text{SWAP}} & \qw & \qw \\
       & \lstick{\ket{\psi}} & \qw & \ghost{\text{SWAP}} & \qw & \qw
    }\]
    \caption{The SWAP test circuit}
    \label{fig:swap}
\end{figure}

It can be shown that the probability the SWAP test accepts the states $\rho$ and $\sigma$ is \cite{kobayashi2003quantum},
\begin{equation}
    \text{Pr}[\text{SWAP accept}] = \frac{1}{2} + \frac{1}{2}\text{Tr}(\rho\sigma)
\end{equation}

In the special case of when at least one of the states (let's say $\sigma$) is a pure state $\sigma = \ket{\psi}\bra{\psi}$, the probability of acceptance is,
\begin{equation}
     \text{Pr}[\text{SWAP accept}] = \frac{1}{2} + \frac{1}{2} |\bra{\psi}\rho\ket{\psi}| = \frac{1}{2} + \frac{1}{2}F^2(\rho, \ket{\psi}\bra{\psi})
     \label{eq:swapaccept}
\end{equation}

Thus when at-least one of the two states is a pure state, the acceptance probability is related to the fidelity between the states. This implies when the states are the same, the probability of acceptance is 1. However, when the states are different then if the SWAP test accepts the states, this implies an error. Thus the error in the SWAP test when the states are different (also called the one-sided error) is the accept probability of the SWAP test while the states are not equal. This error can, however, be brought down to any desired error $\epsilon > 0$ by running multiple instances of the SWAP test circuit. The number of instances required to bring down the error probability to a desired $\epsilon$ is,

\begin{equation*}
\begin{split}
    \text{Pr}[\text{SWAP error}] & = \prod^{M}_{j=1}\text{Pr}[\text{SWAP accept}]_j = (\frac{1}{2} + \frac{1}{2}F^2)^M = \epsilon \\
    & \Rightarrow M(\log(1+F^2)-1) = \log(\epsilon) \Rightarrow M\approx \mathcal{O}(\log(1/\epsilon))
\end{split}
\end{equation*}
where $F = F(\rho, \ket{\psi}\bra{\psi}) = \sqrt{\bra{\psi}\rho\ket{\psi}}$ and we use the fact that fidelity is independent of $\epsilon$. 

\subsubsection{Generalised SWAP test}\label{sec:gswap}
The above SWAP test is optimal in Equality testing (in a single instance) of two unknown quantum states when one has a single copy of the two states. However, there are certain quantum protocols where one has access to multiple copies of one unknown state $\ket{\psi}$ and only a single copy of the other unknown state $\rho$ and the objective is to provide an optimal Equality testing circuit. Considering this scenario, Chabaud et al.~\cite{chabaud2018optimal} provided an efficient construction of such a circuit, generalised SWAP (GSWAP) test circuit. A GSWAP circuit takes as an input a single copy of $\rho$, M copies of $\ket{\psi}$ and $\ceil[\big]{\log M+1}$ copies of the ancilla qubit $\ket{0}$. The generalised circuit is then run on the inputs, and the ancilla qubits are measured in the computational bases. Figure~\ref{fig:gswap} is a generic illustration of such a circuit. For more details on the circuit refer to the original work \cite{chabaud2018optimal}.
\begin{figure}[h!]
\includegraphics[scale=0.30]{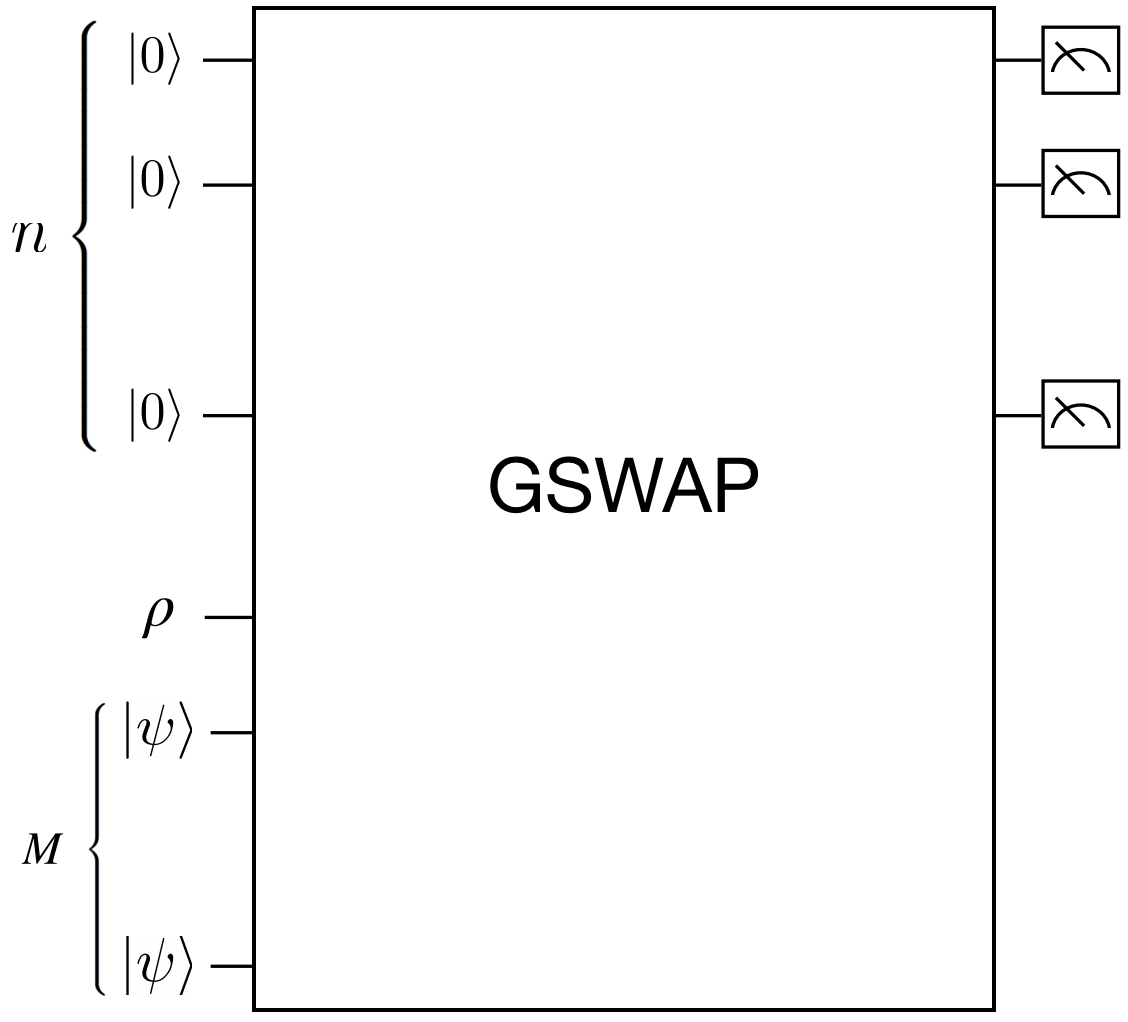}
    \centering
    \caption{GSWAP: A generalisation of the SWAP test with a single copy of $\rho$ and $M$ copies of $\ket{\psi}$. The circuit also inputs $n = \ceil[\big]{\log M+1}$ ancilla qubits in the state $\ket{0}$. At the end of the circuit, the ancilla states are measured in the computational basis.}
    \label{fig:gswap}
\end{figure}
It can be shown that the probability the GWAP circuit accepts two quantum states $\rho$ and $\ket{\psi}$ is,
\begin{equation}
     \text{Pr}[\text{GSWAP accept}] = \frac{1}{M+1} + \frac{M}{M+1} \bra{\psi}\rho\ket{\psi} = \frac{1}{M+1} + \frac{M}{M+1}F^2
     \label{eq:gswap}
\end{equation}
where $F = F(\rho, \ket{\psi}\bra{\psi})$. We note that in the special case of $M=1$, the GSWAP test reduces to the SWAP test. Also in a single instance, GSWAP provides a better Equality test compared to the SWAP test since it reduces the one-sided error probability. In the limit $M \rightarrow \infty$, we obtain the optimal acceptance probability of $\text{Pr}[\text{accept}] = F^2 = |\bra{\psi}\rho\ket{\psi}|$. Another important feature of GSWAP is that it can achieve any desired success probability $\epsilon (\geqslant F^2)$ in just a single instance which is impossible to achieve using SWAP circuit. However, the number of copies required is exponentially more than the number of instances that the SWAP circuit has to run to achieve the same error probability,
\begin{equation}
\begin{split}
    \text{Pr}[\text{GSWAP error}] & = \text{Pr}[\text{GSWAP accept}] = \frac{1}{M+1} + \frac{M}{M+1}F^2 = \epsilon \\
    & \Rightarrow M\approx \mathcal{O}(1/\epsilon)
\end{split}
\label{eq:gswaperror}
\end{equation}
Hence one decides the use of either SWAP test or GSWAP test depending on the specific application.

\subsection{Quantum Physical Unclonable Functions}\label{sec:prelim-qpuf}
A Quantum Physical Unclonable Function, or qPUF, is a secure hardware cryptographic device that is, by assumption, hard to clone or reproduce and utilises properties of quantum mechanics \cite{arapinis2021quantum}. Similar to a classical PUF \cite{armknecht2016towards}, a qPUF is assessed via CRPs. However, in contrast to a classical PUF where the CRPs are classical states, the qPUF CRPs are quantum states. We use the definition of qPUF introduced in~\cite{arapinis2021quantum} for the purpose of this paper. 

A qPUF manufacturing process involves a quantum generation algorithm, `qGen', which takes as an input a security parameter $\lambda$ and generates a PUF with a unique identifier \textbf{id},

\begin{equation}
    \text{qPUF}_{\textbf{id}} \leftarrow \text{qGen}(\lambda)
\end{equation}

Next we define the mapping provided by $\text{qPUF}_{\textbf{id}}$ which takes any input quantum state $\rho_{in} \in \Sc(\mathcal{H}^{d_{in}})$ to the output state $\rho_{out} \in \Sc(\mathcal{H}^{d_{out}})$. Here $\mathcal{H}^{d_{in}}$ and $\mathcal{H}^{d_{out}}$ are the input and output Hilbert spaces respectively corresponding to  the mapping that $\text{qPUF}_{\textbf{id}}$ provides. This process is captured by the `qEval' algorithm which takes as an input a unique $\text{qPUF}_{\textbf{id}}$ device and the state $\rho_{in}$ and produces the state $\rho_{out}$,

\begin{equation}
    \rho_{out} \leftarrow \text{qEval}(\text{qPUF}_{\textbf{id}}, \rho_{in})
\end{equation}

A qPUF needs to satisfy a few requirements. The first property, \textbf{$\delta_r$-Robustness}~\cite{arapinis2021quantum}, ensures that if the qPUF is queried separately with two input quantum states $\rho_{in}$ and $\sigma_{in}$ that are $\delta_r$-indistinguishable to each other, then the output quantum states $\rho_{out}$ and $\sigma_{out}$ must also be $\delta_r$-indistinguishable. The second property, \textbf{$\delta_c$-Collision resistance}~\cite{arapinis2021quantum}, ensures that if the same qPUF is queried separately with two input quantum states $\rho_{in}$ and $\sigma_{in}$ that are $\delta_c$-distinguishable, then the output states $\rho_{out}$ and $\sigma_{out}$ must also be $\delta_c$-distinguishable with an overwhelmingly high probability. Here, the distinguishability is defined with respect to fidelity such that two quantum states $\rho$ and $\sigma$ are $\delta$-distinguishable if $0 \leqslant F(\rho, \sigma) \leqslant  1 - \delta$, where $F(\rho, \sigma)$ is the Uhlmann's fidelity. Alternatively, other distance measures such as trace norm, euclidean norm (any Schatten p-norm) can also be used to define security requirements for qPUF.

The last requirement, which we will use in this paper is the \textbf{$\delta_u$-Uniqueness}~\cite{arapinis2021quantum}. This property ensures that the generation process of qPUF can generate sufficiently distinguishable qPUFs. This is captured by modeling each qPUF
as a quantum operation characterised by a CPTP map that takes the input quantum states in $\mathcal{H}^{d_{in}}$ to output states in $\mathcal{H}^{d_{out}}$. We say that two such maps $\Lambda^{qPUF}_{i}$ and $\Lambda^{qPUF}_{j}$ are $\delta_u$ distinguishable if
\begin{equation}\label{eq:unique}
    Pr[\parallel \Lambda^{qPUF}_{i} - \Lambda^{qPUF}_{j}\parallel_{\diamond} \geq \delta_u[i\neq j]] \geq 1 - \epsilon(\lambda)
\end{equation}

Where $\parallel .\parallel_{\diamond}$ is the diamond norm distance measure for the distinguishability of two quantum operations, and $\epsilon(\lambda)$ is a negligible function in the security parameter $\lambda$. 

The diamond norm is a distance metric for any two completely positive
trace preserving quantum operations $\Lambda_1$, $\Lambda_2$. It is defined as,
\begin{equation}
    \parallel \Lambda_1 - \Lambda_2 \parallel_{\diamond} = \underset{\rho}{\max} (\parallel(\Lambda_1 \otimes \mathbb{I})[\rho] - (\Lambda_2 \otimes \mathbb{I})[\rho]\parallel_1) 
\end{equation}
Operationally it quantifies the maximum probability of distinguishing operation $\Lambda_1$ from $\Lambda_2$ in a single-use.

It has been shown in~\cite{arapinis2021quantum} that unitary maps and $\epsilon$-close to unitary channels, under certain additional conditions, can be considered as a qPUF. We restate the following theorem from~\cite{arapinis2021quantum}:
\begin{theorem}[from~\cite{arapinis2021quantum}]\label{theorem:non-unitary}
Let $\E(\rho)$ be a completely positive and trace-preserving (CPT) map described as follows:
\begin{equation}\label{eq:non-unitary-puf}
    \E(\rho) = (1-\epsilon)U \rho U^{\dagger} + \epsilon \Tilde{\E}(\rho) 
\end{equation}
where $U$ is a unitary transformation, $\Tilde{\E}$ is an arbitrary (non-negligibly) contractive channel and $0 \leq \epsilon \leq 1$. Then $\E(\rho)$ is a ($\lambda,\delta_r,\delta_c$)-qPUF for any $\lambda$, $\delta_r$, and $\delta_c$ and with the same dimension of domain and range Hilbert space, if and only if $\epsilon = \negl(\lambda)$.
\end{theorem}
This is because the properties of robustness and collision resistance can be satisfied by an almost unitary map as a subclass of all CPTP qPUFs.
The uniqueness property on the other hand, is quite challenging and \cite{arapinis2021quantum,kumar2021efficient} showed that one can achieve uniqueness if one samples a unitary from a Haar random set of unitaries. Further, \cite{kumar2021efficient} numerically showed that one can achieve uniqueness if one sample the unitary from a unitary $t$-design set. In this work, we show that one can achieve this property by sampling from a PRU set. Here we consider the qPUF construction to be a unitary matrix $U \in \mathbb{C}^{d \times d}$, where $d = d_{in} = d_{out}$.

A crucial security feature of the qPUF device is its unforgeability property. The unforgeability for qPUFs as a quantum primitive is captured by Definition~\ref{def:qunf}.

It has also been shown in~\cite{arapinis2021quantum} that even though qPUFs cannot satisfy a general existential unforgeability, which is a strong notion for capturing the unpredictability of such hardware, all unitary qPUFs that satisfy the notion of unknownness can satisfy the notion of universal unforgeability. This general possibility result is the consequence of the following theorem proved against any QPT adversary:
\begin{theorem}[restated from~\cite{arapinis2021quantum}]\label{th:sel-qCM-fid}
For any unitary qPUF characterised, and any non-zero acceptance threshold $\delta$ in the fidelity, the success probability of any QPT adversary $\A$ in the universal unforgeability game is bounded as follows:
\begin{equation}
    Pr[1\leftarrow G(\lambda, \A)] \leq \frac{\tilde{d}+1}{d}
\end{equation}
where $d$ is the dimension of the domain Hilbert space, and $0\leq \tilde{d} \leq d-1$ is the dimension of the largest subspace of $\Hild$ that the adversary can span in the learning phase of the Game~\ref{game:uni-unf}.
\end{theorem}
This possibility result, has been later used in~\cite{doosti2020client} to prove security of qPUF-based identification protocols.

\section{Efficient Unforgeability with PRS}\label{sec:unf-prs}
In this section, we investigate the problem of \textit{universal unforgeability} with efficiently producible pseudorandom quantum states. As specified in the Game~\ref{game:uni-unf}, the challenge states need to be picked at random from Haar measure by the challenger. This is an important condition for the unforgeability of unknown unitary transformations. Nevertheless, producing Haar random state is a challenging and resource-intensive task. Hence to take the first step towards the realization of universally unforgeable schemes, we attempt to replace this condition with its computational equivalent, i.e. the notion of PRS, introduced in the preliminary. We first relax this condition by defining a variation of the universal unforgeability game, namely \emph{Efficient Universal Unforgeability} where the challenger picks the challenge states from a pseudorandom family of quantum states. Then we formally prove that unknown unitaries satisfy this notion of unforgeability. Furthermore, we discuss how such pseudorandom quantum states can be efficiently generated using classical pseudorandom functions.

We define the \emph{Efficient Universal Unforgeability} as bellow:

\begin{definition}[Efficent Quantum Universal Unforgeability]\label{def:qunf-efficient}
Let Game $G_{eqUnf}$ be same as Game~\ref{game:uni-unf} except that in the challenge phase, the challenge states are being picked from PRS family of states with a generation algorithm $G(k)$ with a key $k\in\K$, realised in the setup phase. A primitive provides efficient quantum universal unforgeability if the success probability of any QPT adversary $\A$ in winning the $G_{eqUnf}$ is negligible in the security parameter $\lambda$,
\begin{equation}
Pr[1\leftarrow G_{eqUnf}(\lambda, \A)] = \negl(\lambda)
\end{equation}
\end{definition}

Now, for simplicity in the proof, we also define the pseudorandomness property of the PRS with a game as formalized in the following:

\begin{game}\label{game:prs}[PRS distinguishability game] Let $\Hil$ be a Hilbert space and $\K$ the key space. The dimension of $\Hil$ and size of $\K$ depend on the security parameter $\lambda$. Let $\{\ket{\phi_k}\in S(\Hil)\}_{k\in\K}$ be a keyed family of quantum states with efficient generation algorithm $G(k) = \ket{\phi_k}$ on input $k$. We define the following distinguishability game between an adversary $\A$ and a challenger $\C$:
\begin{itemize}
    \item [] \textbf{Setup phase.} The challenger $\C$ selects $k \overset{\$}{\leftarrow} \K$ and $b \overset{\$}{\leftarrow} \{0,1\}$ at random.
    \item [] \textbf{Challenge phase.}
        \begin{itemize}
            \item If $b = 0$ (PRS world): $\C$ prepares $m$ copies of $\ket{\phi^0} = \ket{\phi_k}$ by running $G(k)$. 
            \item If $b = 1$ (Random world): $\C$ prepares $m$ copies of a Haar-random state $\ket{\phi^1} = \ket{\psi}$.
            \item $\C$ sends $\ket{\phi^b}^{\otimes m}$ to $\A$.
        \end{itemize} 
    \item [] \textbf{Guess phase.} $\A$ guesses $b$.
\end{itemize}
\end{game}

Now we establish our main result regarding the efficient unforgeability of unknown unitary primitives. 

\begin{theorem}\label{th:efficientuu-prs}
Any unitary transformation U selected from an unknown unitary family according to Definition~\ref{def:uu}, satisfies efficient universal unforgeability against QPT adversaries. 
\end{theorem}

\begin{proof}
We prove by contraposition in a game-based setting. We want to show that starting from the assumption of pseudorandomness of PRS in the efficient universal unforgeability game, if there exists a QPT adversary who succeeds to win this game, with non-negligible probability, there will also exist an adversary who can efficiently distinguish between PRS and Haar random states, which is in contrast with the initial assumption and as a result show a contradiction. First, we need to specify the following games:
\begin{itemize}
    \item Game 1: This is the universal unforgeability game as specified in Game~\ref{game:uni-unf}, with the only difference that the challenge state $\rho^* = \ket{\phi_{k^*}}\bra{\phi_{k^*}}$ is chosen from a PRS family.
    \item Game 2: This is the PRS distinguishability game as specified in Game~\ref{game:prs}. 
    \item Game 3: This is a variation of Game~\ref{game:prs} where $\C$ in addition to initial resources, has also access to a publicly known and implementable unitary $U$. In the challenge phase, $\C$ does the following: Generates $m$ copies of $\ket{\phi^0} = \ket{\phi_k}$ using $G(k)$, or $m$ copies of Haar random states $\ket{\phi^1} = \ket{\psi}$ depending of $b$, then on each copies applies the public unitary $U$ and sends $(U\ket{\phi^b})^{\otimes m}$ to $\A$. The rest of the game is similar to Game 2.  
    \item Game 4: This game is similar to Game 3, except that $\C$ publicly chooses an $l$ and $l'$ such that $l + l' = m$ and sends $l$ copies of the generated state and $l'$ copies of the state after applying the unitary $U$, i.e. sends $\ket{\phi^b}^{\otimes l} \otimes (U\ket{\phi^b})^{\otimes l'}$ to $\A$.  
    \item Game 5: This game is similar to Game 4 except the public unitary has been replaced by an unknown unitary $\tilde{U}$ of the same dimension. Hence in this game, similar to Game~\ref{def:prs}, we also assume a learning phase for $\A$ before the challenge phase. The learning phase is as follows: $\A$ issues $q=poly(\lambda)$ queries $\{\rho_i\}^q_{i=1}$ to $\C$, on each query $\C$ generates $\rho^{out}_i = \tilde{U}\rho_i\tilde{U}^{\dagger}$ by applying the unitary on the query state and sends $\rho^{out}_i$ to $\A$. Then the rest of the Game is similar to Game 4 and at the end of the challenge phase $\A$ receives $\ket{\phi^b}^{\otimes l} \otimes (\tilde{U}\ket{\phi^b})^{\otimes l'}$
\end{itemize}

\begin{figure}[h!]
\includegraphics[scale=0.45]{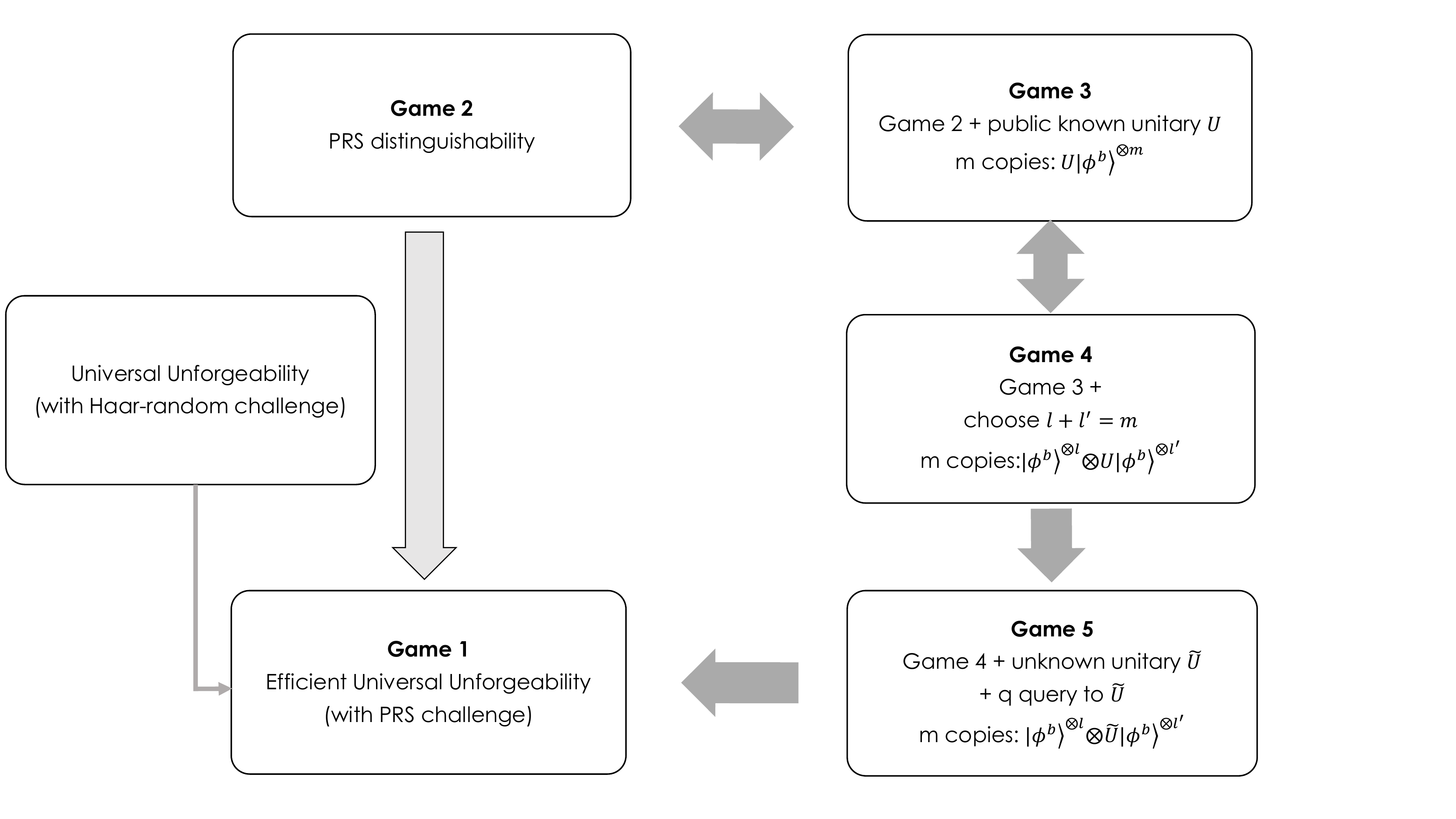}
    \centering
    \caption{Proof sketch of Theorem~\ref{th:efficientuu-prs} with the intermediate games.}
    \label{fig:prs-unf-proof}
\end{figure}

Figure~\ref{fig:prs-unf-proof} illustrates the sketch of the proof. We first show that Game 2, Game 3 and Game 4 are equivalent. We note that unitary transformations are distance invariant and hence they also preserve the distribution of states, as a result applying a unitary to the state will not affect the distribution and the distinguishability of the quantum states, and as a result Game 2 and 3 are equivalent. Furthermore, in Game 4, since the unitary is public, $\A$ can either apply $U$ on the first $l$ copies $\ket{\phi^b}^{\otimes l}$ and end up with $m$ copies of $(U \ket{\phi^b})^{\otimes m}$ or alternatively apply $U^{\dagger}$ on the next $l'$ copies $(U\ket{\phi^b})^{\otimes l'}$ and get $m$ copies of $\ket{\phi^b}^{\otimes m}$, and hence be reduced to either Game 2 or Game 3. As a result, we have
\begin{equation}
    \text{Game 2} \equiv \text{Game 3} \equiv \text{Game 4}
\end{equation}
Now we show that Game 4 implies Game 5 i.e. if an adversary wins the distinguishability in Game 5 with a probability $p$, she will also win in Game 4 with the same probability.

The proof is straightforward as highlighted here. Let $\A$ be an adversary who wins Game 5, which means after the learning phase leading to a polynomial-size database of the input-outputs of the unknown unitary $\tilde{U}$, and receiving $\ket{\phi^b}^{\otimes l} \otimes (\tilde{U}\ket{\phi^b})^{\otimes l'}$, they can guess $b$ with non-negligible probability better than random guess:
\begin{equation}
    \underset{\ket{\phi^b}}{Pr}[b \leftarrow \A(\ket{\phi^b}^{\otimes l} \otimes (\tilde{U}\ket{\phi^b})^{\otimes l'})] = \frac{1}{2} + \nonnegl(\lambda).
\end{equation}
Now let's assume an adversary $\A'$ who plays the Game 4 and has to guess $b$ by receiving the state $\ket{\phi^b}^{\otimes l} \otimes (U\ket{\phi^b})^{\otimes l'}$ can guess $b$ with same $l$ and $l'$ where $U$ is a public unitary. Now $\A'$ can run $\A$ as a subroutine and $\A'$ sends to $\A$ the response to the same learning phase states from $U$. Since $U$ is public $\A'$ can run it locally and produce the required queries. Then $\A'$ also sends the state $\ket{\phi^b}^{\otimes l} \otimes (U\ket{\phi^b})^{\otimes l'}$ to $\A$ and since $\A$ guesses the $b$ with a probability non-negligibly better than half, so does $\A'$. As a result, we have shown that:
\begin{equation}
    \text{Game 4} \Rightarrow \text{Game 5}
\end{equation}
Finally, we show that Game 5 implies Game 1. By contradiction, we assume there exist an adversary $\A$ who wins the unforgeability game with non-negligible probability. Let $\tilde{U}$ be the unknown unitary and $\A$'s forgery state be $\ket{\omega}$ and let the challenge state of Game 1 be a PRS state $\ket{\phi_k}$. We have:
\begin{equation}
\begin{split}
        Pr[1\leftarrow G_{eqUnf}(\lambda, \A)] & = \underset{k}{Pr}[1 \leftarrow \T(\ket{\omega}, (\tilde{U}\ket{\phi_k})^{\otimes \kappa})] \\
        & = \underset{k}{Pr}[F(\ket{\omega}, \tilde{U}\ket{\phi_k}) = \nonnegl(\lambda)] \\
        & = \nonnegl(\lambda).
\end{split}
\end{equation}
Now we construct an adversary $\A'$ playing an instance of Game 5 where $l = 1$ and $l' = m - 1$. In the learning phase $\A$ interacts with the unknown unitary $\tilde{U}$ with the same learning phase states required for $\A$ and sends the query states $\{\rho^{out}_i\}^q_{i=1}$ together with the challenge state $\ket{\phi^b}$ to $\A$. Then $\A$ produces the forgery $\ket{\omega}$ as their guess for $\tilde{U}\ket{\phi^b}$. Now $\A'$ verifies $\ket{\omega}$ with the same test algorithm $\T$ where $\kappa = m - 1$, since $\A'$ has $m-1$ copies of $\tilde{U}\ket{\phi^b}$ to check with. Then $\A'$ outputs the same $b$ as outputted by the $\T$. The success probability of $\A'$ is as follows. If $b=0$, the state is a PRS and the contradiction assumption is satisfied. Hence $\A$'s forgery state will pass the test algorithm with high probability. On the other hand if $b=1$, the state has been picked from Haar measure and as a result of Theorem~\ref{th:uu-unforge}, the success probability of $\A$ winning the forgery game and producing a state to pass the test is negligible. Since guessing $b$ in Game 5 with probability better than random guess is equivalent to the difference between the success probability of $\A'$ in winning the game in the two different scenarios, we have:
\begin{equation}
\begin{split}
        & |\underset{k \leftarrow \K}{Pr}[\A'(\ket{\phi_k} \otimes (\tilde{U}\ket{\phi_k})^{\otimes m-1})=1] - \underset{\ket{\psi} \leftarrow \mu}{Pr}[\A'(\ket{\psi} \otimes (\tilde{U}\ket{\psi})^{\otimes m-1})=1]| \\
        & = |\underset{k \leftarrow \K}{Pr}[\A(\ket{\phi_k})=1] - \underset{\ket{\psi} \leftarrow \mu}{Pr}[\A(\ket{\psi})=1]| \\
        & = \nonnegl(\lambda) - \negl(\lambda) = \nonnegl(\lambda)
\end{split}
\end{equation}
Here, as a specific example, we can consider the GSWAP to be the equality test and, we show how this check can efficiently be performed to show the gap and hence the implication of the two later games. Let us denote the adversary's purified forgery state as $\ket{\omega_b}$. According to equation~(\ref{eq:gswap}), the probability of the GSWAP accepting this state given $m-1$ copies of reference state $\tilde{U}\ket{\phi^b}$, has the following relation with the fidelity of the forgery state:
\begin{equation}\label{eq:gswap-attack-test}
     \text{Pr}[\text{GSWAP accept}] = \frac{1}{m} + \frac{m-1}{m} F(\tilde{U}\ket{\phi^b}, \ket{\omega_b})^2
\end{equation}
Assuming $\A$ wins the unforgeability game for PRS state with non-negligible probability implies that this fidelity is a non-negligible value in the security parameter, hence $F(\tilde{U}\ket{\phi^0}, \ket{\omega_0}) = \delta = \nonnegl(\lambda)$. On the other hand, for Haar-random state this fidelity is always a negligible value and we have that $F(\tilde{U}\ket{\phi^1}, \ket{\omega_1}) = \negl(\lambda)$. As a result the difference between $\A$'s success probability in the two cases is as follows:
\begin{equation}
\begin{split}
        & |\underset{k \leftarrow \K}{Pr}[\A'(\ket{\phi_k} \otimes (\tilde{U}\ket{\phi_k})^{\otimes m-1})=1] - \underset{\ket{\psi} \leftarrow \mu}{Pr}[\A'(\ket{\psi} \otimes (\tilde{U}\ket{\psi})^{\otimes m-1})=1]| \\
        & = \frac{1}{m} + \frac{m-1}{m}F(\tilde{U}\ket{\phi^0}, \ket{\omega_0}) - \frac{1}{m} + \frac{m-1}{m}F(\tilde{U}\ket{\phi^1}, \ket{\omega_1}) \\
        & = \frac{m-1}{m}(\delta - \negl(\lambda)) \approx \frac{m-1}{m}\delta = \nonnegl(\lambda)
\end{split}
\end{equation}

As a result, we have shown that there exist a non-negligible gap and hence $\A'$ can also win the Game 5. In conclusion, we have shown the following relation:
\begin{equation}
    \text{Game 2} \equiv \text{Game 3} \equiv \text{Game 4} \Rightarrow \text{Game 5} \Rightarrow \text{Game 1}
\end{equation}
This means that an adversary winning the unforgeability game, with the challenge being picked from a PRS family, can also distinguish the PRS states from Haar random state which is a contradiction and it concludes the proof. 
\end{proof}

We have formally shown that PRS states are enough to achieve quantum universal unforgeability. The next question is that how such states can be constructed. Ji, Liu, and Song~\cite{ji2018pseudorandom} propose several constructions for generating a PRS family using classical quantum-secure PRFs. Hence they show that PRS can be constructed under the assumption that quantum-secure one-way function exists. Another similar notion called Asymptotically Random State (ARS) has also been introduced in~\cite{brakerski2019pseudo}. In both works, first, oracle access to a classical random function is given to efficiently construct a PRS that is indistinguishable to Haar random states even for exponential adversaries. Then by relying on the existence of quantum-secure one-way function they replace the random function with a post-quantum secure PRF to achieve security against polynomial adversaries. With this approach, one can construct computationally secure $n$-qubit PRS which is also desired for unforgeability security property. Nevertheless, as discussed in~\cite{brakerski2020scalable}, these methods are not scalable and also an n-qubit PRS generator cannot necessarily be used to produce a random state for k-qubit where $k < n$. For these reasons, in~\cite{brakerski2020scalable} the authors introduce a scalable construction for PRS which, unlike prior works, relies on randomising the amplitudes of the states instead of the phase. The authors use Gaussian sampling methods to efficiently achieve PRS.

\section{From Pseudorandom Unitaries to Unknown Unitaries}\label{sec:unf-pru-tdesign}
We prove that a family of unitaries satisfying the computational assumption of PRU, is also a family of unknown unitary transformation. As a result of this implication, efficient constructions such as PRU or t-design can also satisfy the notion of universal unforgeability. Moreover, this result establishes for the first time, a link between a computational assumption of PRU with a hardware assumption such as unknownness.

\begin{theorem}\label{th:pru-uu}
A family of PRU, $\Uset = \{U_k\}_{k \in \K}$ is also a family of unknown unitary (UU) with respect to Definition~\ref{def:uu}.
\end{theorem}
\begin{proof}
We prove this by contradiction. Let $\Uset$ be a family of PRU but not a family of UU which means that there is a quantum polynomial-time (QPT) adversary $\A$ who can estimate the output of a randomly picked $U \leftarrow \Uset$ where $\Uset$ is a UU, on a state $\ket{\psi}$, non-negligibly better than the output of a $U \leftarrow \mu$ picked from a Haar-random unitaries $\mu$ over a $d$-dimensional Hilbert space. Thus for $\A$ the following holds: 
\begin{equation}
\begin{split}
        & |\underset{U \leftarrow \Uset}{Pr}[F(\A(\ket{\psi}),U\ket{\psi}) \geq \nonnegl(\lambda)] - \underset{U_{\mu} \leftarrow \mu}{Pr}[F(\A(\ket{\psi}),U_{\mu}\ket{\psi}) \geq \nonnegl(\lambda)]| \\
        & = \nonnegl(\lambda).
\end{split}
\end{equation}
Let $\A$' be a QPT adversary who aims to break the pseudorandomness property of $\Uset$ using $\A$, and works as follows:\\
\textit{$\A$' picks $\ket\psi$ as one of her chosen inputs in the learning phase of the pseudorandomness game. Then $\A$' also runs $\A$ internally on $\ket{\psi}$.}\\
From the previous equation we know that $\A$ can estimate the output of $U \ket{\psi}$ better than $\U_{\mu}\ket{\psi}$ where $\U_{\mu}$ is a Haar random unitary, by a non-negligible value. Also by definition, we know that The probability that any QPT algorithm estimates the output of any Haar randomly given unitary, is negligible, as the response maps to any random state in the Hilbert space $\Hild$ with exponential distribution~\cite{dankert2006c,nielsen2010quantum}. Thus the equation implies that:
\begin{equation}
        |\underset{U \leftarrow \Uset}{Pr}[F(\A(\ket{\psi}),U\ket{\psi}) \geq \nonnegl(\lambda)]| = \nonnegl(\lambda).
\end{equation}
This means that $\A$ can estimate the output with non-negligible fidelity if the $U$ had been picked from the family. Now $\A$' runs a quantum equality test on the $U \ket{\psi}$ obtained in the learning phase and $\A(\ket{\psi})$. In the case where $U$ is picked from the PRU family, the estimated output and the real output have non-negligible fidelity and the test returns equality with a non-negligible probability. Otherwise, the test shows that they are not equal and $\A$' can conclude that the unitary has been picked from Haar unitaries. Thus for $\A$' we have:
\begin{equation}
    \underset{U \leftarrow \Uset}{Pr}[\A'^{U}(1^{\lambda})=1] - \underset{U_{\mu} \leftarrow \mu}{Pr}[\A'^{U_{\mu}}(1^{\lambda})=1]=\nonnegl(\lambda)
\end{equation}
Therefore we conclude the contradiction.
\end{proof}

We have shown that PRU implies unknown unitaries and followed by the results of~\cite{arapinis2021quantum} for unforgeability of UU, we conclude that PRU makes a set of universally unforgeable unitaries. Now we show that PRU can also be considered as a PUF family. In order to do that we need to show that the PUF requirements discussed in Preliminary~\ref{sec:prelim-qpuf} are satisfied. Since the $\delta_r$-Robustness and $\delta_c$-Collision Resistance are trivially satisfied by the unitarity, we only need to argue the $\delta_u$-Uniqueness requirement. 

\begin{theorem}\label{th:pru-unique}
Let $\Uset \in U(d) = \{U_k\}_{k \in \K}$ be a family of PRU and universally-unforgeable unitary matrices. Then there exist a $\delta_u = \nonnegl(\lambda) = \nonnegl(polylog(d))$ such that $\Uset$ satisfies $\delta_u$-Uniqueness.
\end{theorem}
\begin{proof}
We prove by contraposition and we assume that a non-negligible $\delta_u$ to satisfy the $\delta_u$-Uniqueness does not exist. This means that for any two unitary $U_i$ and $U_j$ picked uniformly at random from $\Uset$, the two unitary are $\zeta$-close in the diamond norm with a high probability. Otherwise if there exist a minimum $\zeta_{min} = \nonnegl(\lambda)$ distance in diamond norm between any two unitaries we have already shown the $\delta_u$ exists. Hence we assume that we have the following condition:
\begin{equation}\label{eq:proof-contra-diamond}
    Pr[\parallel (U_i - U_j)_{i\neq j}\parallel_\diamond \leq \zeta] \geq 1 - \epsilon(\lambda)
\end{equation}
where both $\zeta$ and $\epsilon(\lambda)$ are negligible function in the security parameter. Now we assume an adversary $\A$ wants to distinguish between $\Uset$ and the set of Haar-random unitaries. By assumption, we have that all the unitaries in $\Uset$ are universally unforgeable. So now we let the $\A$ play the PRU game (similar to Game~\ref{game:prs}) while running the universal unforgeability game as a distinguishing subroutine. Let $\C$ be the honest party picking at random a bit $b \in \{0,1\}$ where if $b=0$, a unitary $U$ is picked at random from $\Uset$ and we are in the PRU world and otherwise $U$ is picked from $\mu$ that denotes the set of Haar-random unitary matrices. Then $\A$ gets polynomial oracle access to the $U$ and after the interaction, needs to guess $b$. Now, since there exist an efficient public generation algorithm $Q$ for the PRU set, we let the adversary sample another unitary $U'$ from $Q$ locally and uniformly at random. According to the contraposition assumption give in Equation~\ref{eq:proof-contra-diamond}, if $b=0$, with a high probability these two unitaries are $\zeta$-close in the diamond norm, i.e. $\parallel (U - U')\parallel_\diamond \leq \zeta$. Given this promise, the adversary performs the following strategy: $\A$ locally plays the universal unforgeability game on the $U$, by picking a state $\ket{\psi}$ uniformly at random from Haar measure and querying it to $\C$ as a part of the polynomial oracle interaction with $U$. $\A$ will receive $U\ket{\psi}$ and can ask for multiple copies of it so long as the total number of queries to the oracle remains polynomial. Now we also rely on the fact that since PRU has the efficient computation property, meaning that $\A$ can locally compute the $U'\ket{\psi}$ to get multiple copies. Now $\A$'s strategy to win the unforgeability game is to output $U'\ket{\psi}$ as the forgery for $\ket{\psi}$.

Again in the case of $b=0$, since the two unitaries are negligibly close in the diamond norm with a high probability we have the following:
\begin{equation}
    Pr[\parallel (U - U')\parallel_\diamond \leq \zeta] \geq 1 - \epsilon \Rightarrow Pr[F(U\ket{\psi}, U'\ket{\psi}) \geq 1 - \zeta] \geq 1 - \epsilon
\end{equation}
This holds since the diamond norm is defined as a maximum over all of the density matrices, hence if the two unitaries are very close in the diamond norm, their output over a random state is also very close on average. Thus, the adversary can run a local efficient verification test (for instance a GSWAP test) between $U'\ket{\psi}$ and $U\ket{\psi}$ and use the output of the test as a distinguisher between pseudorandom and Haar-random world. If $b=0$, we have:
\begin{equation}
    Pr[F(U\ket{\psi}, U'\ket{\psi}) \geq 1 - \zeta] \geq 1 - \epsilon \Rightarrow Pr[1\leftarrow G(\lambda, \A)] = \nonnegl(\lambda)
\end{equation}
Hence $\A$ will win the game with a high probability. However, in the case of $b=1$ where $U$ is a Haar-random unitary, we can use lemma 16 in~\cite{kretschmer2021quantum}, that states for a fixed state $\ket{\phi} \in \Hild$ and a Haar-random state $\ket{\psi} \leftarrow \mu$, and any $\epsilon > 0$ we have:
\begin{equation}
    \underset{\ket{\psi} \leftarrow \mu}{Pr}[|\mbraket{\phi}{\psi}|^2 \geq \epsilon] \leq e^{-\epsilon d}
\end{equation}
This implies that taking the $U'\ket{\psi} = \ket{\phi}$ to be the fixed state, we denote that since $U$ is a Haar-random unitary then $U\ket{\psi}$ is also a Haar-random state and hence the probability that the fidelity $F(U\ket{\psi}, U'\ket{\psi})$ is a non-negligible value (with respect to $polylog(d)$) like $ 1 - \zeta$ is exponentially low. Hence in case $b=1$, the probability that the adversary's state passes the verification is exponentially low. Hence using this strategy, there will be always a distinguisher that can distinguish between $\Uset$ and Haar-random unitaries i.e.:
\begin{equation}
    \underset{U \leftarrow \Uset}{Pr}[\A'^{U}(1^{\lambda})=1] - \underset{U_{\mu} \leftarrow \mu}{Pr}[\A'^{U_{\mu}}(1^{\lambda})=1]=\nonnegl(\lambda)
\end{equation}
But this is in contrast with the assumption that $\Uset$ is a PRU. Hence we have reached a contradiction and the proof is complete.
\end{proof}

\section{Pseudorandom Unitaries and States from Hardware Assumptions}\label{sec:composable}
As discussed earlier pseudorandom quantum states can be constructed under the assumption of qPRF or quantum one-way functions. Given the relationship that we have explored in the previous section between the unforgeability of qPUF and quantum pseudorandomness, here we ask whether it is possible to construct pseudorandom quantum states under a different set of assumptions? In this section, we discuss how one can achieve PRU and PRS under hardware assumptions on a family of unitary transformations. These hardware assumptions are generally discussed in the context of quantum PUFs, nevertheless, our results can be in general applied to any sets of unitaries with the given properties as long as they can be assumed on a hardware level.

Let $\Uset \subseteq U(d) = \{U_i\}^{\K}_{i=1}$ be a family of unitaries with certain specific assumption that is given by their physical nature. We want to use the above family as a PRU family or generators for PRS. As shown in~\cite{ji2018pseudorandom}, if $\Uset$ is a PRU then it is also a generators for PRS i.e. $G(k) = U_k\ket{0} = \ket{\phi_k}$. To this end, we investigate the properties of a qPUF family that can be used to achieve pseudorandomness. In the last section, we have shown that PRU implies the notion of unknown unitary assumption, or in other words single-shot unknownness. Now we explore the relation of PRU and another notion of unknownness called \emph{practical unknownness} by Kumar et al.~\cite{kumar2021efficient}. This definition is a more suited definition for t-design unitary sets constructions and is defined as follows:

\begin{definition}[$\epsilon,t,d-$ Practical unkownnness~\cite{kumar2021efficient}]\label{def:pu} We say a unitary transformations $U$, from a set $\Uset \subseteq U(d)$ is $(\epsilon,t,d)$- practically unknown if provided a bounded number $t \leq poly(\log_2 d)$ of queries $U\rho U^{\dagger}$, for any $\rho \in \Hild$, the probability that any $poly(\log_2 d)$-time adversary can perfectly distinguish $U$ from a Haar distributed unitary is upper bounded by $1/2(1 + 0.5 \epsilon)$. Here $0 < \epsilon < 1$, t are functions of $log_2 d$, and $\lim_{\log_2(d) \rightarrow \infty}\epsilon = 0$.
\end{definition}

For the sake of our proof, we need a variation of this definition which is for any polynomial number of queries in the security parameter:

\begin{definition}[$\epsilon,d-$ Practical unkownnness]\label{def:pu-poly} We say a unitary transformations $U$, from a set $\Uset \subseteq U(d)$ is $(\epsilon,d)$- practically unknown if it is $(\epsilon,t,d)$- practically unknown for any $t=poly(\lambda) = poly(\log d)$.
\end{definition}

Now we show that the assumption of $\epsilon,d-$ Practical unkownness implies PRU.

\begin{theorem}\label{th:practicaluu-pru}
A family of $(\epsilon,d)$- practically unknown unitaries where $\epsilon = \negl(\lambda)$ is a PRU family.
\end{theorem}
\begin{proof}
We prove this by contraposition. Let $\Uset = \{U_k\}^{\K}_{i=1} \subseteq U(d)$ be a $(\epsilon,d)$- practically unknown family, that is not a PRU. This means that there exists a QPT adversary $\A$ for which we have the following after some $q=poly(\lambda) = poly(\log(d))$ queries to the unitary oracle:
\begin{equation}
    |\underset{k \leftarrow \K}{Pr}[\A^{U_k}(1^{\lambda})=1] - \underset{U \leftarrow \mu}{Pr}[\A^U(1^{\lambda})=1]| = \delta = \nonnegl(\lambda).
\end{equation}
Equivalently, we can say that if a unitary is randomly picked from either of the sets $\Uset$ or a set of Haar-random distributed unitaries with a random bit $b$, the advantage of the adversary in guessing bit $b$ is a non-negligible function $\delta$ greater than $\frac{1}{2}$. Now if such adversary exists, there exists also an adversary $\A'$ that querying the same $q$ states, can distinguish the $U_k \in \Uset$ from a Haar-random unitary with the following probability:
\begin{equation}
    Pr[\text{distinguish } U_k] \geq \frac{1}{2} + \delta
\end{equation}
On the other hand, if $\Uset$ is $(\epsilon,d)$-practically unknown this probability is equal to $\frac{1}{2}(1 + 0.5 \epsilon)$ where $\frac{\epsilon}{4}$ is a negligible function while as $\delta$ is non-negligible. Hence we reach a contradiction and the proof is complete.  
\end{proof}

We have shown that given the hardware assumption of practical unknownness, over a set of unitary transformations such as unitary qPUFs, one can get PRU and as a result generate PRS by applying random elements of the set on the computational basis state. Now, we want to look at another property of a family of qPUFs and see whether pseudorandomness can be achieved under other related assumptions of such families. One of the main requirements on a qPUF family is the uniqueness property that ensures any two qPUFs in the family are sufficiently distinguishable in the diamond norm. The uniqueness property is formally defined in preliminary section~\ref{sec:prelim-qpuf}, equation (\ref{eq:unique}). In what follows we show a family of unknown and maximally distinguishable unitary matrices, such as unitary qPUFs, also form a family of PRU and are a generator for PRS. 

\begin{theorem}\label{th:max-unique-pru}
Let $\Uset_{\K} = \{U_k\}^{\K}_{k=1} \subseteq U(d)$ be a family of unitary transformation selected at random from a distribution $\chi_{\Uset}$ such that they satisfy almost maximal uniqueness i.e. for any randomly picked pairs of unitary matrices from $\Uset_{\K}$, we have $\parallel (U_i - U_j)_{i\neq j}\parallel_\diamond = 2 - \epsilon$ where $\epsilon = \negl(\lambda)$, then for a sufficiently large $\K$ and $d$, the $\Uset_{\K}$ is also a PRU.
\end{theorem}
\begin{proof}
We first show that if the maximum uniqueness is on average satisfied for any pairs of unitary matrices of $\Uset_{\K}$, then the distribution $\chi_{\Uset}$ converges to Haar measure in the limits of large $d$. The first part of our proof is in the spirit of a proof given in~\cite{kumar2021efficient} for proving uniqueness of Haar-random unitaries. We attempt to prove the other direction for a specific degree of uniqueness which is $2 - \epsilon$ where the maximum of the diamond norm is 2. We have,
\begin{equation}
    \parallel (U_i - U_j)_{i\neq j}\parallel_\diamond = 2 - \epsilon = 2\sqrt{1 - \delta(U_i^{\dagger}U_j)^2}
\end{equation}
Where the $\delta(M) = \underset{\ket{\phi}}{min}|\bra{\phi}M\ket{\phi}|$ is the minimum of absolute value over the numerical range of the operator $M$. From the above equation we have:
\begin{equation}
     \delta(U_i^{\dagger}U_j)^2 = \epsilon - \frac{\epsilon^2}{4} \approx 0
\end{equation}
Since the diamond norm is unitary invariant, we can multiply all the unitaries of the family by a fixed unitary matrix which results in the set including the identity matrix $\mathcal{I}$, hence the above equation can be rewritten as:
\begin{equation}
     \delta(U'_k)^2 = \epsilon - \frac{\epsilon^2}{4}
\end{equation}
where the set of unitary matrices $U'$ is equivalent to the initial set up to a unitary transformation. Now let $\{e^{i\theta_1},\dots,e^{i\theta_d}\}$ be the eigenvalues of $U'_k$. The eigenvalues of a unitary matrix lie on a unit circle $\mathbb{S}^1 \subset \mathcal{C}$. As shown in the Lemma 1.1 of~\cite{kumar2021efficient}, the following relation exist between the distribution of the eigenvalues of a general unitary matrix in an arc of size $\theta$, and the function $\delta(U)$:
\begin{equation}
     \delta(U'_k)^2 = \frac{1}{2} + \frac{1}{2} \cos{\theta}
\end{equation}
Where $\theta = \theta_j - \theta_k$ for pairs of eigenvalues $\{e^{i\theta_j},e^{i\theta_k}\}$. From the above equation we have:
\begin{equation}
     \theta = \theta_j - \theta_k = \arccos{(-1 + 2\epsilon - \frac{\epsilon^2}{2})} \approx \pi - \sqrt{\epsilon} + \dots
\end{equation}
Now we can use Theorem~\ref{th:wieand}. Let $N_{\theta}$ be a random variable that represents the number of eigenvalues in an arc of size $\theta$. Then we have the expectation value of this random variable for the given distribution where the $\theta = \pi - \epsilon'$, and $\epsilon' = \negl(\lambda)$, to be
\begin{equation}
     \mathbb{E}_d[N_{\theta}] = \frac{d\times\theta}{2\pi} = \frac{d}{2} - \frac{\epsilon'd}{2\pi}
\end{equation}
which is close to half of the total number of eigenvalues since the second term is always smaller than 1. This means that in the limit of large $d$, every diameter of the unit circle divide the circle into two areas that each on average includes half of the eigenvalues. Also the variance of the random variable $N_{\theta}$ will be: 
\begin{equation}
     Var(N_{\theta}) = \frac{1}{\pi^2}(\log(d) + 1 + \gamma + \log|2\sin(\frac{\pi - \epsilon'}{2})|) + o(1) \approx \frac{\log(d)}{\pi^2} + c' + o(1)
\end{equation}
where $\gamma \approx 0.577$ and $c' < 1$. Next we calculate the probability that for our given distribution, there are more than half of the eigenvalues in each half of the circle denoted by an arc or size $\pi - \epsilon'$. Using the Chernoff bound we have:
\begin{equation}
     Pr[N_{\pi - \epsilon'} - \mathbb{E}_d[N_{\pi - \epsilon'}]| > x\mathbb{E}_d[N_{\pi - \epsilon'}]] \leq e^{-\frac{x^2}{2+x}\mathbb{E}_d[N_{\pi - \epsilon'}]}
\end{equation}
Here we want the $x\mathbb{E}_d[N_{\pi - \epsilon'}]$ to be equal to $\frac{d}{2}$, so we have $x = \frac{d/2}{d/2 - d\epsilon'/2\pi} = \frac{1}{1 - \epsilon'/\pi}$ and since the $x$ is a small value the above inequality can be used. Substituting this into the above equation we will have:
\begin{equation}
     Pr[N_{\pi - \epsilon'} - \mathbb{E}_d[N_{\pi - \epsilon'}]| > \frac{d}{2}] \leq e^{-\frac{(\frac{1}{1 - \epsilon'/\pi})^2}{2+\frac{1}{1 - \epsilon'/\pi}}\times(d/2 - \epsilon'd/2\pi)} \approx e^{-d/6}
\end{equation}
since $\epsilon'$ is negligible. This shows that with a very high probability, on every half of the unit circle, there exist half of the eigenvalues of the random matrix from our specified distribution. We conclude eigenvalues of a random unitary from the distribution $\chi_{\Uset}$ are uniformly distributed on the unit circle. Let us denote this uniform distribution on $\mathbb{S}^1$ by $\nu$. In order to compare the distribution of $\chi_{\Uset}$ with the Haar measure, we use the empirical spectral measure introduces in the Appendix~\ref{ap:haar}.
We denote the empirical spectral distance of $\chi_{\Uset}$ as $\tilde{\mu}_{\chi}$ and for Haar measure we denote it as $\tilde{\mu}_{H}$. Since we have shown that the eigenvalues of matrices from $\chi_{\Uset}$ are distributed uniformly on $\mathbb{S}^1$, it is easy to see that $\mathbb{E}(\tilde{\mu}_{\chi}) = \nu$ and in the limit of large $d$ we have the convergence in probability $\tilde{\mu} \overset{d\rightarrow \infty}{\longrightarrow} \nu$. Now we use the Theorem~\ref{th:diaconis-shah} (Appendix~\ref{ap:haar}) that implies the convergence of the empirical spectral measure of the set of unitaries picked from Haar measure to $\nu$, in the limit of large $d$. Having the these two convergence and the properties of the limit we can conclude that the empirical spectral measure for $\chi_{\Uset}$ converges to the one for Haar measure. Then we look at Kolmogorov distance of the eigenvalues of these two distributions. We rely on the result given in~\cite{meckes2019sharp} that shows the Kolmogorov distance between the distributions of eigenvalues of random unitary matrices is given by $d_K(\mu, \nu) = \underset{0\leq \theta < 2\pi}{sup} |\frac{N_{\theta}}{d} - \frac{\theta}{2\pi}|$ and specifically for Haar measure it is bounded by 
\begin{equation}
     d_K(\mu_{H}, \nu) \leq c \frac{\log(d)}{d}
\end{equation}
Where $c > 0$ is a universal constant. Given the fact that for the specific value of $\theta$ for the distribution of $\chi_{\Uset}$ the Kolmogorov distance $d_K(\mu_{\chi}, \nu)$ is of the order $\frac{1}{d}$ which is negligible and using the triangle inequality for the Kolmogorov distance we have
\begin{equation}
\begin{split}
         d_K(\mu_{H}, \mu_{\chi}) & \leq d_K(\mu_{H}, \nu) + d_K(\nu, \mu_{\chi})\\
         & \leq c \frac{\log(d)}{d} + \negl(\lambda)\\
         & \leq \negl(\lambda)
\end{split}
\end{equation}
Thus the distribution of the eigenvalues of the random matrices of $\chi_{\Uset}$ is negligibly close to the Haar measure. Also for any randomly picked matrix from each of these distributions, the eigenvalues are fixed. As a result, the convergence between the distribution of the eigenvalues of matrices leads to the fact that in the limit of large $d$, $\chi_{\Uset}$ converges to the Haar measure on the unitary set. 

Finally, we show that a polynomial time quantum adversary given a polynomial query to each unknown unitary $U_k$ cannot distinguish any member of this family from Haar measure. This is straightforward since the two distributions are asymptotically  close. Thus we have:
\begin{equation}
    |\underset{k \leftarrow \K}{Pr}[\A^{U_k}(1^{\lambda})=1] - \underset{U \leftarrow \mu}{Pr}[\A^U(1^{\lambda})=1]| = \negl(\lambda).
\end{equation}
And we have shown that the set $\Uset_{\K}$ is a PRU.
\end{proof}

\section{Efficient Quantum Identification Protocols Using Quantum Pseudorandomness}\label{sec:efficient-qpufid}
We discuss the application of some of our previously established results, in order to achieve an efficient quantum identification protocol. \emph{Identification} (also called Entity authentication), is a method to prove the identity of one party called \emph{prover} to another party called \emph{verifier}. In the quantum setting, either the verifier or the prover or both have some quantum capabilities and the properties of quantum mechanics are used to enhance the security of such protocols against powerful quantum adversaries. Here we focus on two quantum identification protocols, proposed in~\cite{doosti2020client}. These identification protocols are based on quantum PUFs and use their unforgeability property to achieve exponential security against QPT adversary (polynomial time in the learning phase, and unbounded during the quantum communication) in a polynomial number of rounds. Even though these protocols are resource-efficient in many aspects, one of the main practical challenges in implementing these protocols is the fact that in order to use the unforgeability property of quantum PUFs, the challenge states needs to be sampled at random from Haar measure. As a result, relying on Theorem~\ref{th:efficientuu-prs}, we show that these protocols can still achieve exponential security using PRS. This brings us one step closer to the practical implementations of quantum identification protocols with exponential security against powerful quantum adversaries and can lead to promising solutions to the problem of untrusted manufacturers. Furthermore, using Theorem~\ref{th:pru-uu}, we show that PRU can also be used as an alternative to Hardware assumptions in order to run these identification protocols.

First, we briefly mention the two protocols. The full description of both protocols can be found in Appendix~\ref{ap:qpufid-protocols}. 

\subsection{Identification protocol with high-resource verifier}
In this qPUF-based protocol, The verifier uses a database of the challenge-responses of the qPUF in order to identify a party who has access to the qPUF device. Since the verifier needs to run a quantum verification algorithm to check the response-state received by the prover, this protocol is referred to as \emph{high-resource verifier}. The protocol has three phases: \emph{Setup phase}, \emph{Identification phase} and \emph{Verification phase}.

In the \emph{Setup phase} the verifier who has physical access to the qPUF device samples some quantum challenges at random from the Haar measure and records the response state of the qPUF on a quantum database. Then publicly sends the qPUF over to the prover. At this stage, polynomial access to the qPUF device has been assumed for the adversary i.e. the quantum adversary can query the qPUF with a polynomial number of arbitrary quantum states attempting to learn the behaviour of the underlying unitary transformation.

In the \emph{Identification phase}, the verifier picks one of the challenges in the database at random and sends them to the prover over a public quantum channel, while an adversary has full control over the channel. Then the prover who acquires the qPUF obtains the correct response to the challenge states and sends it back through the same public quantum channel. 

Finally, in the \emph{Verification phase}, the verifier needs to verify the challenge state to confirm the identity of the other party. To this end, the verifier needs to run a quantum verification or test algorithm on the received response, and the $M$ copies of the correct response that is stored in the database. In~\cite{doosti2020client} the protocol has been proposed and analysed with both SWAP and GSWAP test as the verification algorithm.

The following theorem states the security or soundness of the above protocol with both of the tests:
\begin{theorem}[\textbf{Th. 2 and 4 in~\cite{doosti2020client}}]\label{th:hr-sound} Let qPUF be a selectively unforgeable\footnote{Universally unforgeable in our terminology} unitary PUF over $\Hild$. The success probability of the adversary to pass the {\normalfont SWAP}-test or {\normalfont GSWAP}-test verification of the high resource verifier protocol is at most $\epsilon$, given that there are $N$ different CRPs, each with $M$ copies. The $\epsilon$ is bounded as follows for each verification:

\begin{equation}
         \text{Pr}[\text{Ver accept}_{\A}] \leqslant \epsilon \quad \quad \epsilon_{\text{SWAP}} \approx \mathcal{O}(\frac{1}{2^{NM}}) \quad \quad \epsilon_{\text{GSWAP}} \approx \mathcal{O}\big(\frac{1}{(M+1)^{N}}\big)
\end{equation}
\end{theorem}

We now introduce a computationally efficient variation of this protocol which we denote as \emph{Efficient hr-verifier identification protocol}, by replacing the qPUF with any general universally unforgeable pseudorandom unitary and the Haar-random challenges with pseudorandom quantum states as follows:

\begin{enumerate}
    \item \emph{Setup phase}:
            \begin{enumerate}
                \item Verifier has access to a PRU family $\Uset \subseteq U(d) = \{U_i\}^\K_{i=1}$.
                \item Verifier samples at random $k \overset{\$}{\leftarrow} \K$ and selects $U_k$.
                \item Verifier has also access to a family of PRS $\{\ket{\phi_{k'}}\in S(\Hild)\}_{k'\in\K'}$ and randomly picks $Q \in \mathcal{O}(\text{poly} \log d)$ of them as the challenge states.
                \item Verifier queries the $U_k$ individually with each challenge $\ket{\phi_{k'}}$ a total of $M$ number of times to obtain $M$ copies of the response state $\ket{\phi^r_{k'}}$ and stores them in their local database $S$. 
                \item The verifier transfers the $U_k$ to Prover or securely sends the key $k$.
            \end{enumerate}
      \item \emph{Identification phase}:
            \begin{enumerate}
                \item Verifier uniformly selects a challenge labelled ($i \xleftarrow{\$} [Q]$), and sends the state $\ket{\phi_i}$ over a public quantum channel to Prover.
                \item Prover generates the output $\ket{\phi^p_i}$ by querying the challenge from $U_k$.
                \item The output state $\ket{\phi^p_i}$ is sent to Veifier over a public quantum channel.
                \item This procedure is repeated with the same or different states a total of $R \leq Q$ times.
            \end{enumerate}
        \item \emph{Verification phase}:
            \begin{enumerate}
                \item Verifier runs a quantum equality test algorithm on the received response from Bob and the $M$ copies of the correct response that she has in the database. This algorithm is run for all the $R$ CRP pairs.
                \item Verifier outputs `1' implying successful identification if the test algorithm returns `1' on all CRPs. Otherwise, outputs `0'. 
            \end{enumerate}
\end{enumerate}

We note that the protocol assumes that the adversary has only query access to the unitary $U_k$ from a PRU family as it is also assumed in the Definition~\ref{def:pru}. The following theorem which is a corollary of the previous results shows that the \emph{Efficient hr-verifier identification protocol} is also exponentially secure against QPT adversary with the same security bounds.

\begin{theorem}\label{th:eff-hr-sound} Let $U_k \in \Uset$ be unitary randomly selected from a PRU family $\Uset \subseteq U(d)$. The success probability of the adversary to pass the {\normalfont SWAP}-test or {\normalfont GSWAP}-test verification of the \emph{Efficient hr-verifier protocol} is at most $\epsilon$, given that there are $N$ different CRPs, each with $M$ copies. The $\epsilon$ is bounded as follows for each verification:
\begin{equation}
         \text{Pr}[\text{Ver accept}_{\A}] \leqslant \epsilon \quad \quad \epsilon_{\text{SWAP}} \approx \mathcal{O}(\frac{1}{2^{NM}}) \quad \quad \epsilon_{\text{GSWAP}} \approx \mathcal{O}\big(\frac{1}{(M+1)^{N}}\big)
\end{equation}
\end{theorem}
\begin{proof}
First, we use Theorem~\ref{th:pru-uu} that shows $\Uset$ is also an unknown unitary family. Then we use Theorem~\ref{th:efficientuu-prs} that states any UU unitary satisfies efficient universal unforgeability which is the universal unforgeability given the states are picked from the PRS family. These two results suggest that the $U_k$ within the protocol satisfies the same notion of universal unforgeability that qPUF satisfies in the original protocol. Now we can directly use the result of~\cite{doosti2020client} using the SWAP and GSWAP test which will result in the same security bound in the number of rounds and copies of challenge-response pairs.
\end{proof}

\subsection{Identification protocol with low-resource verifier}
The second identification protocol also introduced in~\cite{doosti2020client}, enables a weak verifier to identify a quantum server prover in the network. The main idea behind this protocol is to delegate the equality testing to the prover so that the verifier can only run a classical verification algorithm. While it might seem this delegation could damage the security, it has been shown that the unforgeability property of qPUF combined with some trapification techniques used in the protocol leads to yet another exponentially secure qPUF-based identification protocol. In addition to enabling the clients to identify quantum servers on the clouds, this protocol has the advantage of one-way quantum communication compared to the previous protocol. We give a brief description of the protocol here. The complete protocol can be found in Appendix~\ref{ap:low-verifier}.

The \emph{Setup phase} is similar to the previous protocol, except that in addition to the challenge-response pairs, the verifier also generates some trap states. These trap states need to be orthogonal to the challenge subspace.

In the \emph{Identification phase}, the verifier sends two quantum states in every communication rounds. One of the states is the challenge state and the other one is either the correct response or the trap states with no overlap with the correct response. The verifier selects the correct or trap response at random with probability $1/2$\footnote{It has been shown that this probability can be generalised to arbitrary distribution.}. In other words in $N$ rounds, around $N/2$ positions are the states sent with their correct responses.

In the \emph{Verification phase}, the prover generates the valid response for every challenge by interacting them with the qPUF device and then runs a SWAP-test on the response produced by the qPUF and the other state sent by the verifier. The prover then sends the classical output of the test to the verifier who receives a classical $S_N = s_1,...,s_N$ where $s_i \in \{0,1\}$. Finally, the verifier runs a classical verification algorithm on this string that checks the expected result for the positions with the valid responses and also the statistics of the remaining positions. 

The protocol has been proven secure against both collective and coherent attacks with the following bound which we restate from the original paper:

\begin{theorem}[\textbf{Th. 6, 7 and 8 in~\cite{doosti2020client}}]\label{th:lr-sound} Let qPUF be a universally unforgeable unitary PUF over $\Hild$. The success probability of any QPT adversary $\A$ (using coherent or collective strategy) to pass the verification of the low resource verifier protocol is at most $\epsilon$, in $N$ rounds. The $\epsilon$ is of the order $\mathcal{O}(\frac{1}{2^{N}})$
\begin{equation}
         \text{Pr}[\text{Ver accept}_{\A}] \leqslant \epsilon \quad \quad \epsilon \approx \mathcal{O}(\frac{1}{2^{N}})
\end{equation}
\end{theorem}

Similar to the previous protocol, we introduce an efficient version of this protocol by replacing the Haar-random states with PRS and the qPUF with an unknown unitary selected from a PRU family. We denote this protocol as \emph{Efficient lr-verifier identification protocol} and it is described as follows:

\begin{enumerate}
    \item \emph{Setup phase}:
            \begin{enumerate}
                \item Verifier has access to a PRU family $\Uset \subseteq U(d) = \{U_i\}^\K_{i=1}$.
                \item Verifier samples at random $k \overset{\$}{\leftarrow} \K$ and selects $U_k$.
                \item Verifier has also access to a family of PRS $\{\ket{\phi_{k'}}\in S(\Hild)\}_{k'\in\K'}$ and randomly picks $Q \in \mathcal{O}(\text{poly} \log d)$ of them as the challenge states.
                \item Verifier queries the $U_k$ individually with each challenge $\ket{\phi_{k'}}$ a total of $M$ number of times to obtain $M$ copies of the response state $\ket{\phi^r_{k'}}$ and stores them in their local database $S$. 
                \item Verifier selects states $\ket{\phi^{\perp}}$ orthogonal to the selected challenge's subspace and queries the $U_k$ with them to obtain the trap states labelled as $\ket{\phi^{\text{trap}}}$. The unitary property ensures that $\langle \phi^{\text{trap}}|\phi_{k'}^r\rangle = 0$.
                \item The verifier transfers the $U_k$ to Prover or securely sends the key $k$.
            \end{enumerate}
      \item \emph{Identification phase}:
            \begin{enumerate}
                \item Verifier randomly selects a subset $N \subseteq \K'$ different challenges from the database, and sends the state $\ket{\phi_i}$ over a public quantum channel to Prover.
                \item Verifier randomly selects $N/2$ positions, marks them $b = 1$ and sends the valid response states $\ket{\phi_i^1} = \ket{\phi_i^r}$ to Prover. On the remaining $N/2$ positions, marked as $b = 0$, Verifier sends the trap states $\ket{\phi_i^0} = \ket{\phi_i^{\text{trap}}}$.

            \end{enumerate}
        \item \emph{Verification phase}:
            \begin{enumerate}
                \item Prover queries $U_k$ with the challenge states to generate the response states $\ket{\phi^p_i}$ for all $i \in [N]$. 
                \item Prover performs a SWAP test between $\ket{\phi^p_i}$ and the response state $\ket{\phi_i^b}$ received from the Verifier. This algorithm is repeated for all the $N$ distinct challenges.   
                \item Prover labels the outcome of $N$ instances of the SWAP test algorithm by $s_i \in \{0,1\}$ and sends them over a classical channel to Prover.
                \item Verifier runs a classical verification algorithm \texttt{cVer($s_1,...,s_N$)}(as specified in~\cite{doosti2020client} and Appendix~\ref{ap:low-verifier}) and outputs `1' implying that identification has been successful and outputs `0' otherwise.
            \end{enumerate}
\end{enumerate}

Again using the previous proof techniques presented in~\cite{doosti2020client} and our results, we show that the \emph{Efficient lr-verifier identification protocol} satisfies exponential security against QPT adversary both under the coherent and collective attack models.

\begin{theorem}\label{th:eff-lr-sound} Let $U_k \in \Uset$ be unitary randomly selected from a PRU family $\Uset \subseteq U(d)$. The success probability of a QPT adversary $\A$ to pass the verification of the \emph{Efficient lr-verifier protocol} is at most $\epsilon$, in $N$ rounds. The $\epsilon$ is bounded as follows:
\begin{equation}
         \text{Pr}[\text{Ver accept}_{\A}] \leqslant \epsilon \quad \quad \epsilon \approx \mathcal{O}(\frac{1}{2^{N}})
\end{equation}
\end{theorem}
\begin{proof}
First, we specify that we can directly use the result of the Theorem (6) in~\cite{doosti2020client} which bounds the success probability of classical adversary in passing the classical verification algorithm. Then the success probability against a quantum adversary with the collective and coherent attack is defined as the advantage of the quantum adversary over that classical adversary in guessing the trap states, using all the side information obtained from the $U_k$ in the learning phase. 
Using Theorem~\ref{th:pru-uu} we have that $\Uset$ is also an unknown unitary family. Then we use Theorem~\ref{th:efficientuu-prs} that states any UU unitary satisfies efficient universal unforgeability which is the universal unforgeability given the states are picked from the PRS family. Next, the conditions of Theorems (7) and (8) in~\cite{doosti2020client} are satisfied and we can directly use those results that state the success probability of such adversaries in guessing the traps is bounded as follows:
\begin{equation}
    \text{Pr}[b \leftarrow \Lambda_{\A}] \leqslant \frac{1}{2} + \mathcal{O}(2^{-N})
\end{equation}
Where $\Lambda_{\A}$ denotes any map that $\A$ uses to distinguish the traps states.
Finally, putting all the above results together we have
\begin{equation}
\text{Pr}[\text{Ver accept}_{\A}]  \leqslant \epsilon = \text{Pr}[\text{Ver accept}_{\text{Classical Adv}}] + \mathcal{O}(2^{-N}) \approx \mathcal{O}(2^{-N}) 
\end{equation}
This concludes the soundness proof of \emph{Efficient lr-verifier protocol}.
\end{proof}

\section{Conclusion and Discussion}\label{sec:conclusion}
We have explored the relationship between quantum pseudorandomness and quantum hardware assumptions such as quantum physical unclonability. Since one of the main cryptographic properties of quantum physical unclonable functions is the notion of universal unforgeability, we have inspected whether quantum pseudorandomness would be enough as a challenge sampling requirement, to achieve this level of unforgeability. We have formally proved that the answer to this question is positive. This result can improve the practicality of qPUF-based constructions and protocols since it will replace the requirement of Haar-randomness on the challenge states, which is resourceful and experimentally challenging.

We have also established the link between the notions of unknown unitary and PRU. We proved that any family of PRU is also a family of unknown unitaries and, hence they could be a potential candidate for the construction of qPUF devices. This result can be a complement to the result of~\cite{kumar2021efficient} where they show t-designs can also satisfy a similar notion, namely practical unknownnes, which leads to an efficient proposal for constructing quantum PUFs. 

Then we also looked at the problem of generating pseudorandom quantum states from hardware assumptions. Our results show that different physical assumptions that were proposed in the context of PUFs, such as uniqueness or practical unknownnes, can also imply quantum pseudorandomness. This is of theoretical interest as it shows an alternative way for achieving quantum pseudorandomness which is different from current approaches based on post-quantum and computational assumptions. Apart from the cryptography perspective, having a different set of assumptions for PRS and PRU can find potential applications in physics~\cite{bouland2019computational}. Another interesting future direction would be to further explore the relationship between unclonability and quantum pseudorandomness that has been initially proposed in~\cite{ji2018pseudorandom}, relying upon our new results.

Finally, to show the consequence of our result in the practicality of qPUF-based protocols, we have revisited the qPUF-based identification protocols proposed in~\cite{doosti2020client} using PRS and we have shown that this more efficient version of these protocols can achieve the same security guarantee as they were initially proposed. 

An important note that needs to be emphasised regarding these protocols is that they assume during the transition stage (e) in the setup phase, the adversary has only query access to the device. If the PRU is realised by hardware assumptions such as practical unknownness as showed in Theorem~(\ref{th:practicaluu-pru}), then this requirement is satisfied by assumption. Otherwise, the unitary circuit of the $U_k$ selected from the PRU needs to be obfuscated or hidden from the adversary \cite{alagic2016quantum,brakerski2020quantum}. This problem mainly arises if the PRU is built from classical PRF and hence the underlying circuit is publicly known. Another alternative way to go around this problem is that only the key index of the selected unitary is sent in a secure way to the other party who is running the selected unitary locally. Thus the above protocol works naturally with hardware assumptions that imply the unitary transformation is unknown. Nevertheless, using PRU constructions with known unitary circuits has the advantage that removes the quantum memory requirements to store the response pairs, as one can presumably compute the response state having access to the circuit and only store the related classical parameters.

Yet another interesting future direction would be to establish concrete bounds on the randomness and pseudorandomness of unitary families given different degrees of uniqueness or distinguishability (not negligibly close to perfect distinguishability), in terms of the diamond norm. This is also related to the study of t-design unitaries and the toolkit from random matrix theory, which we used in this paper can be potentially beneficial and powerful tools for this study.

\section*{Acknowledgement}
We acknowledge the UK Engineering and Physical Sciences Research Council grant  EP/N003829/1 and as well as Innovate UK  funded project called AirQKD : product of a UK industry pipeline, grant number 106178.

\section*{Competing Interests}
The authors declare no competing interest.

\medskip
 \bibliographystyle{unsrt}
\bibliography{efficient-qunforgeability}

\appendix
\section*{Appendix}
\section{Haar Measure Group and Properties of Random Matrices}\label{ap:haar}
A Haar measure is a non-zero measure on any locally compact group $G$ such that $\mu: G \rightarrow [0,\infty)$ 
such that for all $X \subset G$ and $x \in G$ we have the following translation invariance property for $\mu(X) = \int_{x \in G} d\mu(x)$:
\begin{equation}
    \mu(xX) = \mu(Xx) = \mu(X)
\end{equation}
In particular, the Haar measure $d\mu(U)$ can be defined for a unitary group $U(d)$. Sampling unitaries from Haar measure on $U(d)$ is equivalent to geometrically uniform sampling from unitary groups of a certain dimension. In practice, however, sampling from the Haar measure requires exponential (in $d$) resources~\cite{knill1995approximation}.

In this work, we are interested in characterising the properties of the eigenvalues of Haar-random unitary matrices and their distributions. To this end, we introduce the following important results from the random matrix theory.
The first result that we need, is known as \emph{Weyl density formula} or \emph{Weyl integration formula}, and is stated as follows:

\begin{lemma}[Weyl integration formula on $\mathbb{U}(n)$~\cite{meckes2019random}]
Let $\{e^{i\theta_j}\}^n_{j=1}$ be the eigenvalues of $n\times n$ random unitary matrix. The unordered eigenvalues of a random unitary matrix have the following eigenvalue density
\begin{equation}
    \frac{1}{n!(2\pi)^n} \prod_{1\leq j < k \leq n} |e^{i\theta_j} - e^{i\theta_k}|^2
\end{equation}
with respect to $d\theta_1\dots d\theta_n$ on $(2\pi)^n$. That is, for any $g:\mathbb{U}(n)\rightarrow R$ with
\begin{equation*}
    g(U) = g(VUV^*) \quad \text{for any } U,V \in \mathbb{U}(n), 
\end{equation*}
(i.e., $g$ is a class function), if $U$ is Haar-distributed on $\mathbb{U}(n)$, then
\begin{equation}
    \mathbb{E}[g(U)] = \frac{1}{n!(2\pi)^n} \int_{[0,2\pi)^n} \tilde{g}(\theta_1,\dots,\theta_n) \prod_{1\leq j < k \leq n} |e^{i\theta_j} - e^{i\theta_k}|^2 d\theta_1\dots d\theta_n
\end{equation}
where $\tilde{g}: [0,2\pi)^n \rightarrow \mathbb{R}$ is the (necessarily symmetric) expression of $g(U)$ as a
function of the eigenvalues of $U$.
\end{lemma}

As discussed in~\cite{meckes2019random}, one consequence of the above lemma is that the eigenvalues of random unitary matrices want to spread out. For any given pair of eigenvalues labelled by $(j,k)$, $|e^{i\theta_j} - e^{i\theta_k}|^2$ is zero if $\theta_j = \theta_k$, and is 4 if $\theta_j = \theta_k + \pi$ (and in that neighborhood if they are roughly antipodal). This produces the
effect alternatively known as ``eigenvalue repulsion".

Another important tool in the study of the eigenvalues of random matrices is the \emph{empirical spectral measure} defined as,
\begin{equation}
     \tilde{\mu} = \frac{1}{n}\sum^n_{j=1} \delta_{e^{i\theta_j}}
\end{equation}
where $e^{i\theta_j}$ are the eigenvalues of the unitary matrix and $\delta$ is the probability distribution function over the eigenvalues. The empirical spectral measure is a probability measure to encode the ensemble of eigenvalues which puts equal mass at each of the eigenvalues of $U$. This encoding is very useful for representing the spreading of the eigenvalues on the complex unit circle denoted by $\mathbb{S}^1 \subseteq \mathbb{C}$.

Next, we need the following important theorem by Diaconis-Shashahani~\cite{diaconis1994eigenvalues}, that shows the convergence of the eigenvalues of the Haar-random matrices to the uniform distribution over the unit circle:

\begin{theorem}\label{th:diaconis-shah}
Let $U$ be uniformly chosen from Haar-measure in $U(d)$, Let $\nu$ be the uniform distribution on $\mathbb{S}^1$. Then as $d \rightarrow \infty$, the $\tilde{\mu}_U$ converges, weakly in probability, to $\nu$:
\begin{equation}
    \tilde{\mu}_U \overset{d\rightarrow \infty}{\longrightarrow} \nu
\end{equation}
\end{theorem}

Finally, we use the following result by Wieand~\cite{wieand2002eigenvalue}:

\begin{theorem}\label{th:wieand}
Let $U$ be a unitary matrix chosen from Haar measure in $U(d)$, and let $\{e^{i\theta_1},\dots,e^{i\theta_d}\}$ be the eigenvalues of $U$. Fix a finite number of intervals on the unit circle $I_1 =(e^{i\theta_{1j}} , e^{i\theta_{1l}}),\dots,I_m =(e^{i\theta_{mj}} , e^{i\theta_{ml}})$. Define the random variables $N_{\theta_1},\dots,N_{\theta_m}$ to be the number of eigenvalues in each arc defined by the intervals. In the limit of large $d$, the mean and variance of $N_{\theta_k}$ are as follows:
\begin{equation}
    \mathbb{E}_d[N_{\theta_k}] = \frac{d(\theta_{kj} - \theta_{kl})}{2\pi}
\end{equation}
and 
\begin{equation}
     Var(N_{\theta_k}) = \frac{1}{\pi^2}(\log(d) + 1 + \gamma + \log|2\sin(\frac{\theta_{kj} - \theta_{kl}}{2})|) + o(1).
\end{equation}
where $\gamma \approx 0.577$ is the Euler's constant.
\end{theorem}
This theorem, gives a concrete formula for calculating the expectation value and variance of the random variable that represents the number of eigenvalues of a random unitary matrix, in each arc of the unit circle and hence can be used to study the distribution of eigenvalues of random matrices.

\section{qPUF-based Identification Protocols}\label{ap:qpufid-protocols}
In this appendix we give the full description of the qPUF-based identification protocols introduced in~\cite{doosti2020client} which we briefly describes in Section~\ref{sec:efficient-qpufid}.

\subsection{Identification with high-resource verifier}\label{ap:high-verifier}
This protocol, is run between the Alice, the verifier, and Bob, the prover and it is divided into three sequential phases,\\ \\
\begin{framed}
\begin{enumerate}
    \item \emph{Setup phase}:
            \begin{enumerate}
                \item Alice has the qPUF device. 
                \item She randomly picks $K \in \mathcal{O}(\text{poly} \log D)$ classical strings $\phi_i \in \{0,1\}^{\log D}$.
                \item Alice selects and applies a Haar-random state generator operation denoted by the channel $\E$ to locally create the corresponding quantum states in $\HilD$: $\phi_i \overset{\E}{\rightarrow} \kc,\hspace{2mm} \forall i \in [K]$.
                \item She queries the qPUF individually with each challenge $ \kc$ a total of $M$ number of times to obtain $M$ copies of the response state $\kr$ and stores them in their local database $S \equiv \{\kc, \krm\}_{i=1}^{K}$. 
                \item Alice publicly transfers the qPUF to Bob.
            \end{enumerate}
    \item \emph{Identification phase}:
            \begin{enumerate}
                \item Alice uniformly selects a challenge labelled ($i \xleftarrow{\$} [K]$), and sends the state $\kc$ over a public quantum channel to Bob.
                \item Bob generates the output $\kp$ by querying the challenge received from Alice to the qPUF device.
                \item The output state $\kp$ is sent to Alice over a public quantum channel.
                \item This procedure is repeated with the same or different states a total of $R \leq K$ times. 
            \end{enumerate}

    \item \emph{Verification phase}:
            \begin{enumerate}
                \item Alice runs a quantum equality test algorithm on the received response from Bob and the $M$ copies of the correct response that she has in the database. This algorithm is run for all the $R$ CRP pairs.
                \item She outputs `1' implying successful identification if the test algorithm returns `1' on all CRPs. Otherwise, she outputs `0'. 
            \end{enumerate}
    The quantum verification algorithm run by Alice can be both SWAP or GSWAP tests described in the preliminary section.
\end{enumerate}
\end{framed}

\subsection{Identification with low-resource verifier}\label{ap:low-verifier}
This protocol is run between Alice, the verifier, and Bob, the prover in three sequential phases,
\begin{framed}{\textbf{Low-resource verifier qPUF-based protocol}}

\begin{enumerate}
    \item \emph{Setup phase}:
            \begin{enumerate}
                \item Alice has the qPUF device. 
                \item Alice randomly picks $K \in \mathcal{O}(\text{poly} \log D)$ classical strings $\phi_i \in \{0,1\}^{\log D}$.
                \item Alice selects and applies a Haar-random state generator operation denoted by the channel $\E$ to locally create the corresponding quantum states in $\HilD$: $\phi_i \overset{\E}{\rightarrow} \kc,\hspace{2mm} \forall i \in [K]$.
                \item Alice queries the qPUF individually with each quantum challenge $\kc$ to obtain the response state $\kr$.
                \item Alice creates states $\ket{\phi_i^{\perp}}$ orthogonal to $\kc$ and queries the qPUF device with them to obtain the trap states labelled as $\ket{\phi_i^{\text{trap}}}$. The unitary property of qPUF device ensures that $\langle \phi_i^{\text{trap}}|\phi_i^r\rangle = 0$. 
                \item She creates a local database $S \equiv \{\kc,\{ \kr, \ket{\phi_i^{\text{trap}}}\}\}$ for all $i \in [K]$. Thus the $S$ registers stores the challenge state $\kp$ and the corresponding valid response state and the trap state which is orthogonal to the response state.
                \item Alice publicly transfers the qPUF to Bob.
            \end{enumerate}
            The transition is non-secure and Eve is allowed $\mathcal{O}(\text{poly} \log D)$ query access to the qPUF to build their database. 
    \item \emph{Identification phase}:
            \begin{enumerate}
                \item Alice randomly selects a subset $N \subseteq K$ different challenges $\kc$ and sends them over a public channel to Bob.
                \item She randomly selects $N/2$ positions, marks them $b = 1$ and sends the valid response states $\ket{\phi_i^1} = \kr$ to Bob. On the remaining $N/2$ positions, marked as $b = 0$, she sends the trap states $\ket{\phi_i^0} = \ket{\phi_i^{\text{trap}}}$.
            \end{enumerate}
    \item \emph{Verification phase}:
            \begin{enumerate}
                \item Bob queries the qPUF device with the challenge states received from Alice to generate the response states $\kp$ for all $i \in [N]$. 
                \item He performs a quantum equality test algorithm by performing a SWAP test between $\kp$ and the response state $\ket{\phi_i^b}$ received from Alice. This algorithm is repeated for all the $N$ distinct challenges.   
                \item Bob labels the outcome of $N$ instances of the SWAP test algorithm by $s_i \in \{0,1\}$ and sends them over a classical channel to Alice.
                \item Alice runs a classical verification algorithm \texttt{cVer($s_1,...,s_N$)} and outputs `1' implying that Bob's qPUF device has been successfully identified. She outputs `0' otherwise. 
            \end{enumerate}
\end{enumerate}
\end{framed}

The classical verification algorithm, \texttt{cVer}, receives an $N$-bit binary string $S_N$ as input. The algorithm is divided into two tests as is as follows:
\begin{algorithm}[ht!]
\SetAlgoLined
\textbf{Description:} Let $S_N = \{0,1\}^N$ be the input $N$-bit string. Let $P=\{i_k\}^{N/2}_{k=1}$ be the set of indices showing the rounds of the protocol where $b=1$. Algorithm consists of two tests, \texttt{test1} and \texttt{test2} as follows:\\
   \texttt{test1:}\\
   \ForAll{$i$ in P}{
    \If{$s_{i} = 0$}{
        $count\gets count+1$\;
    }}
    \eIf{$count = \frac{N}{2}$}{
        \Return 1\;
    }{\Return 0\;}

\hrulefill\\
\texttt{test2:}\\
    \eIf{\texttt{test1} = 0}{
        \Return 0\;
    }
    {
    \ForAll{$i$ not in P}{
    \If{$s_{i} = 1$}{
        $count\gets count+1$\;
    }}
    \eIf{$\lvert count - \delta\frac{N}{2} \rvert \leqslant \delta_{er}$}{
        \Return 1\;
    }{\Return 0\;}
    }
 \caption{\texttt{cVer} algorithm}
\end{algorithm}\label{alg:cver}

\end{document}